\documentclass[11pt]{article}
\usepackage{fullpage}
\usepackage{graphicx}
\usepackage{amssymb}
\usepackage{amsmath,amsthm}
\usepackage{color}

\usepackage[table]{xcolor}

\usepackage{enumitem}
\usepackage{ifpdf}
\ifpdf    
\usepackage{hyperref}
\else    
\usepackage[hypertex]{hyperref}
\fi

\newcommand{\remove}[1]{}
\newcommand{\conf}[1]{}

\newtheorem{definition}{Definition}[section]
\newtheorem{theorem}{Theorem}[section]

\newtheorem{lemma}[theorem]{Lemma}

\newtheorem{corollary}[theorem]{Corollary}

\newcommand{\ceil}[1]{\lceil{#1}\rceil}
\newcommand{\poly}{\operatorname{poly}}
\newcommand{\floor}[1]{\left\lfloor{#1}\right\rfloor}

\newcommand{\neweps}{{\gamma}}

\providecommand{\card}[1]{|#1|}

\newcommand{\comment}[1]{}
\newcommand{\junk}[1]{}

\newcommand{\hcm}[1][1]{\hspace*{#1 cm}}
\newcommand{\rb}[2]{\raisebox{#1 mm}[0mm][0mm]{#2}}
\newcommand{\istrut}[2][0]{\rule[- #1 mm]{0mm}{#1 mm}\rule{0mm}{#2 mm}}
\newcommand{\zeromath}[1]{\makebox[0mm][l]{$#1$}}

\newcommand{\E}{{\mathbb E\/}}
\newcommand{\paren}[1]{{\left( #1 \right)}}

\newcommand{\fr}[2]{\mbox{$\frac{#1}{#2}$}}
\newcommand{\modulo}{\operatorname{mod}}
\renewcommand{\div}{\operatorname{div}}

\newcommand{\polylog}{\mathrm{polylog}}

\newcommand{\bucket}{\mbox{\sc bucket}}
\newcommand{\bbucket}{\bucket^\star}

\newcommand{\ThreeSUM}{\textsf{3SUM}}

\newcommand{\OMv}{\textsf{OMv}}
\newcommand{\ConvolutionThreeSUM}{\textsf{Convolution\ThreeSUM}}
\newcommand{\ConvThreeSUM}{\textsf{Conv\ThreeSUM}}

\newcommand{\SetDisjointness}{\textsf{SetDisjointness}}
\newcommand{\SetIntersection}{\textsf{SetIntersection}}

\newcommand{\QES}{\textsf{QES}}

\newcommand{\Bjorklund}{Bj{\o}rklund}

\newcommand{\Patrascu}{P\v{a}tra\c{s}cu}

\begin{document}

\title{Higher Lower Bounds from the 3SUM Conjecture\thanks{Supported by NSF grants
CCF-1217338, CNS-1318294, CCF-1514383, CCF-1637546, CCF-1815316, and a grant from
the US-Israel Binational Science Foundation.
This research was performed in part at the Center for Massive Data
Algorithmics (MADALGO) at Aarhus University,
which is supported by the Danish National Research Foundation grant DNRF84.}}

\author{
Tsvi Kopelowitz\thanks{Email: kopelot@gmail.com} \\ Bar-Ilan University \and
Seth Pettie \thanks{Email: pettie@umich.edu.} \\ University of Michigan \and
Ely Porat \thanks{Email: porately@gmail.com} \\ Bar-Ilan University}

\date{}

\maketitle

\begin{abstract}
The \ThreeSUM{} conjecture has proven to be a valuable tool for proving conditional lower bounds on
dynamic data structures and graph problems.
This line of work was initiated by \Patrascu{} (STOC 2010)
who reduced \ThreeSUM{} to an offline \SetDisjointness{} problem.
However, the reduction introduced by \Patrascu{} suffers from several inefficiencies,
making it difficult to obtain {\em tight} conditional lower bounds from the \ThreeSUM{} conjecture.

In this paper we address many of the deficiencies of \Patrascu's framework.  We give new and efficient
reductions from \ThreeSUM{} to offline \SetDisjointness{} and offline \SetIntersection{} (the reporting version of \SetDisjointness)
which leads to polynomially higher lower bounds for several problems.
Using our reductions, we are able to show the essential optimality of several classic algorithms, assuming the \ThreeSUM{} conjecture.
\begin{itemize}
\item{} Chiba and Nishizeki's $O(m\alpha)$-time algorithm (SICOMP 1985) for enumerating all
triangles in a graph with arboricity/degeneracy $\alpha$ is essentially optimal, for any $\alpha$.
\item{} Bj{\o}rklund, Pagh, Williams, and Zwick's algorithm (ICALP 2014) for listing $t$ triangles is essentially optimal,
assuming the matrix multiplication exponent is $\omega=2$.
\item{} Any static data structure for \SetDisjointness{} that answers queries in constant time must spend $\Omega(N^{2-o(1)})$
time in preprocessing, and any such data structure that spends near-linear preprocessing must spend $\Omega(N^{1/2-o(1)})$
time per query.  Here $N$ is the size of the set system.  Since answering two keyword searches is at least as hard as
\SetDisjointness, these lower bounds imply that any substitute for the \emph{inverted index} used by search engines
cannot simultaneously have ideal preprocessing and query times.
\end{itemize}
These statements were unattainable via \Patrascu's reductions.

We also introduce several new reductions from \ThreeSUM{} to pattern matching problems and dynamic graph problems.
Of particular interest is a new conditional lower bound on Dynamic Maximum Cardinality Matching,
which uses a new technique for obtaining \emph{amortized} lower bounds.\\

\centerline{\emph{This paper is dedicated to the memory of Mihai \Patrascu}}
\end{abstract}

\section{Introduction}\label{sect:introduction}

Data structure lower bounds come in two varieties: conditional and unconditional. The strongest unconditional lower bounds
(in the cell probe model) are either poly-logarithmic~\cite{PatrascuD06,Larsen12,Yu16,LarsenWY18} or
on extreme tradeoffs between sub-logarithmic update time
and large query time; see~\cite{PT11a,AlmanWY18}.
\Patrascu~\cite{Patrascu10} proposed an approach for proving polynomial {\em conditional} lower bounds (CLBs) based on the
conjectured hardness of the \ThreeSUM{} problem, via a new intermediate problem he called \ConvolutionThreeSUM{}.
The integer versions of the \ThreeSUM{} and \ConvolutionThreeSUM{} problems are defined as follows.
Given a set $A \subset \mathbb{Z}$,  the \ThreeSUM{} problem is
to decide if there is a triple $(a,b,c) \in A^3$ of distinct elements such that $a+b=c$. The \ConvolutionThreeSUM{} problem is,
given a vector $A \in \mathbb{Z}^n$, to decide if there is a pair $(i,j)\in[n]^2$ for which $A(i)+A(j) = A(i+j)$.
Here $[n] = \{0,1,\ldots,n-1\}$ is the first $n$ naturals.

It was originally conjectured~\cite{GajentaanO95} that \ThreeSUM{} requires $\Omega(n^2)$ time.  However, there are now known
to be $O(n^2/\polylog(n))$ \ThreeSUM{} algorithms for both integer~\cite{BaranDP08} and
real~\cite{GronlundP14,GronlundP18,GoldS17,Freund17,Chan18} inputs.
The {\em modern \ThreeSUM{} Conjecture} is that the time complexity of the problem is $\Omega(n^{2-o(1)})$, even in expectation.

\paragraph{\ThreeSUM{} and Set Disjointness.}
\Patrascu~\cite{Patrascu10} gave a reduction from \ThreeSUM{} to \ConvolutionThreeSUM{}
showing that the \ThreeSUM{} conjecture implies that \ConvolutionThreeSUM{} also requires $\Omega(n^{2-o(1)})$ time.
This reduction is summarized in the following theorem.  Observe that the optimum value for $k=k(n)$ depends on
the actual complexities of \ThreeSUM{} and \ConvolutionThreeSUM.

\begin{theorem}[\Patrascu~\cite{Patrascu10}]\label{thm:I3S_2_IC3S_mihai}
Define $T_{3S}(n)$ and $T_{C3S}(n)$ to be the randomized (Las Vegas) complexities of
\ThreeSUM{} and \ConvolutionThreeSUM{} on instances of size $n$.
For any parameter $k$, $T_{3S}(n) = O(n^2/k + (k^3 + k^2\log n)\cdot T_{C3S}(n/k))$.
\end{theorem}

\Patrascu{} then reduced \ConvolutionThreeSUM{} to the offline \SetDisjointness{} problem.
The input to this problem consists of a universe $U$ of elements,
a family $\mathcal F \subset 2^U$ of subsets of $U$,
and $q$ pairs of subsets $(S,S')\in \mathcal{F}\times \mathcal{F}$.
For each query pair we are interested in whether $S\cap S' = \emptyset$ or not.
The following two Theorems summarize \Patrascu's reductions.
(The second Theorem is implicit in Section 2.3 of~\cite{Patrascu10}.)

\begin{theorem}[\Patrascu~\cite{Patrascu10}]\label{thm:mihai_disjoint}
Let $g(n)$ be such that \ConvolutionThreeSUM{} requires $\Omega(\frac{n^{2}}{g(n)})$ expected time.
For any $\sqrt{n\cdot g(n)} \ll n^{\epsilon} \ll n/g(n)$
let $\mathbb{A}$ be an algorithm for the offline \SetDisjointness{} problem where $|U| =n$, $|\mathcal{F}|= \Theta(n^{1/2+\epsilon})$, each set in $\mathcal{F}$ has at most $O(n^{1-\epsilon})$ elements from $U$, each element in $U$ appears in
$\Theta(\sqrt n)$ sets from $\mathcal{F}$, and
$q=\Theta(n^{1+\epsilon})$. Then $\mathbb{A}$ requires $\Omega(\frac{n^{2}}{g(n)})$ expected time.
\end{theorem}

\begin{theorem}[\Patrascu~\cite{Patrascu10}]\label{thm:mihai_disjoint_extended}
Let $g(n)$ be such that \ConvolutionThreeSUM{} requires $\Omega(\frac{n^{2}}{g(n)})$ expected time.
For any $\sqrt{n\cdot g(n)} \ll n^{\epsilon} \ll n/g(n)$
let $\mathbb{A}$ be an algorithm for the offline \SetDisjointness{} problem where $|U| =\Theta(n^{2-2\epsilon})$, $|\mathcal{F}|= \Theta(n^{1/2+\epsilon}\log n)$, each set in $\mathcal{F}$ has at most $O(n^{1-\epsilon})$ elements from $U$, each element in $U$ appears in $\Theta(n^{2\epsilon -1/2})$ sets from $\mathcal{F}$, and $q=\Theta(n^{1+\epsilon}\log n)$. Then $\mathbb{A}$ requires $\Omega(\frac{n^{2}}{g(n)})$ expected time.
\end{theorem}


The \ThreeSUM{} conjecture implies that $g(n) = n^{o(1)}$, and so $\epsilon \in (1/2,1)$.
Using this case, \Patrascu{} provided CLBs for triangle enumeration, set-disjointness,
and various other dynamic graph and data structure problems.

\paragraph{Related work.}
\Patrascu's results led to a surge of research on CLBs that assume the \ThreeSUM{} conjecture. Vassilevska Williams and Williams~\cite{WW13} showed that finding a triangle of zero weight in a weighted graph requires $\Omega(n^{3-o(1)})$ time.
Abboud, Vassilevska Williams, and Weimann~\cite{AWW14} proved that the Local Alignment problem, which is of great importance in computational biology, cannot be solved in truly sub-quadratic time.
Amir, Chan, Lewenstein, and Lewenstein~\cite{ACLL14} proved that for the Jumbled Indexing problem on a text with $n$ integers (from a large enough alphabet), either the preprocessing time needs to be truly quadratic or the query time needs to be truly linear.
Abboud and Vassilevska Williams~\cite{AW14} showed numerous CLBs conditioned on various popular conjectures.
They showed, in particular, that conditioned on the \ThreeSUM{} conjecture,
data structure versions of $(s,t)$-Reachability, Strong Connectivity, Subgraph Connectivity,
Bipartite Perfect Matching, and ``Pagh's problem'' all require non-trivial polynomial preprocessing, query, or update time.
Abboud, Vassilevska Williams, and Yu~\cite{AWY15} introduced a different approach for
reducing \ThreeSUM{} to dynamic graph problems via intermediate problems called
$\Delta$-Matching Triangles and Triangle Collection.

After the conference version of this paper was published there has been plenty of followup work that uses our reductions
or extends our results in new directions.  Amir et al.~\cite{AKLPPS16} proved CLBs for pattern matching
with \emph{gaps} using our triangle enumeration results.
Kopelowitz and Krauthgamer proved CLBs for various distance oracles on colored points~\cite{KK16}.
Goldstein, Kopelowitz, Lewenstein, and Porat~\cite{GKLP16} proved CLBs
for partial matrix multiplication and problems related to finding witnesses for convolutions.
In~\cite{GKLP17}, the same authors considered a \emph{data structural} version of the \ThreeSUM{} conjecture,
and derived time/space tradeoffs from it.  Dahlgaard~\cite{Dahlgaard16} extended our lower bounds
on incremental maximum cardinality matching to other problems, and other hardness assumptions.

\subsection{Limitations of \Patrascu's Reductions}

\Patrascu's reductions are ingenious but suffer from a few limitations, namely:
\begin{itemize}
\item The number of \SetDisjointness{} queries is rather large (at least $\omega(n^{3/2})$)
making it impossible to obtain any $\omega(n^{1/2})$ time lower bound per query.
\item The size $N$ of the set systems in the offline \SetDisjointness{} problem is also rather large
($N = \Omega(n^{3/2})$) so it is impossible to get lower bounds of $\Omega(N^{4/3})$ time as a function of $N$.
\item The sets in the offline \SetDisjointness{} instances are very sparse.  In \Patrascu's framework
all sets have size at most $O(\sqrt{|U|})$, thereby limiting the possibility of obtaining
meaningful CLBs for the dense case of many graph problems, e.g., triangle enumeration.
\item Finally, \Patrascu's reduction from \ThreeSUM{} to offline \SetDisjointness{} is only meaningful if the true complexity of \ThreeSUM{} is $\omega(n^{13/7})$.\footnote{To see this, note that Theorem~\ref{thm:mihai_disjoint} is only applicable
if $g(n) = O(n^{1/3})$, i.e., \ConvolutionThreeSUM{} requires $\Omega(n^{5/3})$ time.  It is consistent with
Theorem~\ref{thm:I3S_2_IC3S_mihai} that \ConvolutionThreeSUM{} can be solved in
$O(n^{5/3})$ time while \ThreeSUM{} can be solved in $O(n^{13/7})$ time.
Furthermore, even if \ThreeSUM{} were proved to be $\Omega(n^{5/3})$ unconditionally,
this would imply no superlinear lower bound on \ConvolutionThreeSUM{} (via Theorem~\ref{thm:I3S_2_IC3S_mihai})
or to any of the problems that \ConvolutionThreeSUM{} is reduced to.}
\end{itemize}

Regarding the first two limitations, we would like to have much more control over the size of the set system
and the number of queries.  The third limitation seems intrinsic to {\em any} reduction to \SetDisjointness{}
where the sets are essentially random, since, by the birthday paradox, random $\omega(\sqrt{|U|})$-size sets will
almost surely intersect, giving no useful information per query.
The fourth limitation gets at the issue of {\em robustness}: how valuable are reductions from \ThreeSUM{} if the
\ThreeSUM{} conjecture turns out to be false---but not by much?  The possibility of a truly subquadratic \ThreeSUM{} algorithm
seemed remote a few years ago but given recent developments~\cite{Williams14,GronlundP14,GronlundP18,CL15,KaneLM18} it is not so absurd.
Chan and Lewenstein~\cite{CL15} showed that numerous special cases of
\ThreeSUM{} can be solved in truly subquadratic time.

\subsection{New Result: A More Versatile Lower Bound Framework}

\begin{figure*}
\centering\begin{tabular}{lllll}
\multicolumn{5}{c}{\sc Reductions to Set Disjointness/Intersection}\\
\istrut[2]{6}Reduction & $|\mathcal{F}|$ & $|U|$ & $q$ & Remarks\\\cline{1-5}
\istrut[2]{5}\ConvThreeSUM{} $\rightarrow$ \SetDisjointness & $n^{1/2 +
\epsilon}$ & $n$ & $n^{1+\epsilon}$ & $\sqrt{n g(n)} \ll n^\epsilon
\ll n/g(n)$\\
\istrut[2]{5}\ConvThreeSUM{} $\rightarrow$ \SetDisjointness & $n^{1/2 +
\epsilon} \log n$ & $n^{2-2\epsilon}$ & $n^{1+\epsilon}\log n$&\\\hline
\istrut[2]{5}\ThreeSUM{} $\rightarrow$ \SetDisjointness & $n\log n$ &
$n^{2-2\gamma}$ & $n^{1+\gamma}\log n$ & $\gamma \in [0,1)$ and $\delta\in(0,1)$\\
\istrut[2]{5}\ThreeSUM{} $\rightarrow$ \SetIntersection &
$n^{\frac{1}{2}(1+\delta+\gamma)}$ & $n^{1+\delta-\gamma}$ &
$n^{1+\gamma}$ & Total output size: $O(n^{2-\delta})$\\\cline{1-5}
\end{tabular}
\caption{\label{fig:3sumreductions}The first two reductions from \ConvolutionThreeSUM{} are from \Patrascu~\cite{Patrascu10}.
The third and fourth reductions from \ThreeSUM{} are new.}
\end{figure*}

We overcome the limitations of \Patrascu's framework by giving efficient
reductions directly from \ThreeSUM{} to \SetDisjointness{}
and from \ThreeSUM{} to \SetIntersection, which we define shortly.
By avoiding \ConvolutionThreeSUM{} as an intermediate problem,
our reductions imply non-trivial lower bounds,
even if the true complexity of \ThreeSUM{} is just
$\Omega(n^{3/2 + \epsilon})$, for any $\epsilon > 0$.

For clarity's sake, we present our new framework in terms of a new constant
$\gamma$ instead of the constant $\epsilon$ of Theorem~\ref{thm:mihai_disjoint}.
Notice that in \Patrascu's frameworks $1/2<\epsilon<1$ while here $0< \gamma < 1$.
The first theorem reduces \ThreeSUM{} to  offline \SetDisjointness.

\begin{theorem}\label{thm:improved_reduction}
Let $f(n)$ be such that \ThreeSUM{} requires expected time $\Omega(\frac{n^{2}}{f(n)})$.
For any constant $0< \gamma < 1$ 
let $\mathbb{A}$ be an algorithm for  offline \SetDisjointness{} where $|U| = \Theta(n^{2-2\gamma})$, $|\mathcal{F}|=\Theta(n\log n)$, each set in $\mathcal{F}$ has at most $O(n^{1-\gamma})$ elements from $U$, and $q=\Theta(n^{1+\gamma}\log n)$. Then $\mathbb{A}$ requires $\Omega(\frac{n^{2}}{f(n)})$ expected time.
\end{theorem}

We also consider the offline \SetIntersection{} problem where the input is the same as in \SetDisjointness,
except that we are required to enumerate a large enough subset of all elements in the $q$ intersections.
Note that in this case we allow $\gamma = 0$.
Refer to Section~\ref{section:improved_reduction} for proofs of Theorems~\ref{thm:improved_reduction} and \ref{thm:improved_reduction_reporting}.

\begin{theorem}\label{thm:improved_reduction_reporting}
Let $f(n)$ be such that \ThreeSUM{} requires expected time $\Omega(\frac{n^{2}}{f(n)})$.
For any constants $0\leq \gamma < 1$
and $\delta>0$, let $\mathbb{A}$ be an algorithm for offline \SetIntersection{}
where $|U| = \Theta(n^{1+\delta-\gamma})$, $|\mathcal{F}|=\Theta(\sqrt{n^{1+\delta+\gamma}})$,
each set in $|\mathcal{F}|$ has at most $O(n^{1-\gamma})$ elements from $U$,
$q=\Theta(n^{1+\gamma})$, and the algorithm is required to enumerate at most $O(n^{2-\delta})$ elements from the intersections.
Then $\mathbb{A}$ requires $\Omega(\frac{n^{2}}{f(n)})$ expected time.
\end{theorem}

These theorems eliminate the need to use \ConvolutionThreeSUM{} as a stepping stone (Theorem~\ref{thm:I3S_2_IC3S_mihai})
when proving lower bounds via \SetDisjointness/\SetIntersection.
Nonetheless, the \ConvolutionThreeSUM{} problem is still useful and its exact relationship to the \ThreeSUM{} problem is an interesting open question.
The more constrained structure of \ConvolutionThreeSUM{} makes it easier to use in some CLBs; see~\cite{ACLL14,WW13}.
We are also able to prove
that the randomized complexities of \ThreeSUM{} and \ConvolutionThreeSUM{} differ by at most a logarithmic factor.
Refer to Section~\ref{app:proof_conv3sum} for proof of Theorem~\ref{thm:conv3sum}.

\begin{theorem}\label{thm:conv3sum}
Define $T_{3S}(n)$ and $T_{C3S}(n)$ to be the randomized (Las Vegas) complexities of
\ThreeSUM{} and \ConvolutionThreeSUM{} on instances of size $n$.
Then $T_{3S}(n) = O(\log n \cdot T_{C3S}(n))$.
\end{theorem}

\subsection{Techniques and Implications}

Our new framework borrows liberally from the techniques of \Patrascu's framework, particularly the use of almost linear hash functions.
However, in order to obtain our improvements we assemble the building blocks in a new and simpler way.
The implications of our new framework are significant. To start off, since $q$ can be made arbitrarily close to $n$ (by having $\gamma$ approach 0), it is now possible to obtain much higher CLBs for the query time of
\SetDisjointness{}. Similarly, since the size of the offline \SetDisjointness{} can be made
small (by having $\gamma$ approach 0), it is now possible to obtain higher CLBs in terms of the size of the offline \SetDisjointness{} instance.

Finally, as is illustrated in Section~\ref{sec:triangle_lb}, using the new framework it is now possible to obtain CLBs for graph problems which apply to all edge densities. Such reductions make use of Theorem~\ref{thm:improved_reduction_reporting}.
For example, in the case of triangle enumeration our CLBs imply that clever techniques, such as
fast matrix multiplication, cannot lead to faster enumeration algorithms, even in very dense graphs.

To illustrate the advantages of the new framework, we briefly show how we are able to obtain higher CLBs for the fundamental problem of {\em online} \SetDisjointness{}, and then continue in the body of this paper to describe new and better CLBs for many old and new problems.
A corresponding discussion on CLBs for the online \SetIntersection{} problem is given in Section~\ref{app:details_set_intersection}.

\subsection{A Higher Conditional Lower Bound for Online \SetDisjointness{}}

It is straightforward to see that \emph{online} \SetDisjointness{} solves \emph{offline} \SetDisjointness.
We phrase the CLBs in terms of $N$: the sum of set sizes.
Define $t_p$ to be the preprocessing time of an online \SetDisjointness{} structure
and $t_q$ its query time.

Using \Patrascu{'s} reduction from \ConvolutionThreeSUM{} to \SetDisjointness{} we have $N=\Theta(n^{1.5})$ and
there are $\Theta(n^{1+\epsilon}) = \Theta(N^{(2 +2\epsilon)/3})$ queries that need to be answered.
Thus we obtain the following lower bound tradeoff: $t_p + N^{(2 +2\epsilon)/3}\cdot t_q= \Omega \left(\frac{N^{4/3}}{g(N^{2/3})}\right)$.
The \ThreeSUM{} conjecture implies $g(x) = x^{o(1)}$.
If, for example, we only allow linear preprocessing, letting $\epsilon$ tend to $1/2$ gives
a query lower bound of $\Omega(N^{\frac{1}{3}-o(1)})$. If we demand constant time queries,
we obtain a lower bound of $\Omega(N^{\frac 4 3 -o(1)})$ on the preprocessing time.
We show next how our new framework provides better tradeoffs.

\begin{theorem}\label{thm:set_disjoint_improved}
Assume the \ThreeSUM{} conjecture.
For any $0< \gamma <1$, any data structure for \SetDisjointness{} has

\[
t_p + N^{\frac{1+\gamma}{2-\gamma}}\cdot t_q = \Omega\left(N^{\frac{2}{2-\gamma}-o(1)}\right).
\]

\end{theorem}


\begin{proof}
Using Theorem~\ref{thm:improved_reduction}, we have $N=\Theta( n^{2-\gamma}\log n) $, and the number of queries to answer is $\Theta(n^{1+\gamma}\log n) = \tilde{\Theta}(N^{\frac{1+\gamma}{2-\gamma}})$.
By the \ThreeSUM{} conjecture, answering these queries
takes time $\Omega(n^{2-o(1)}) = \Omega(N^{\frac{2}{2-\gamma} - o(1)})$.
\end{proof}

Theorem~\ref{thm:set_disjoint_improved} implies, for example, that if we only allow linear preprocessing time,
then by making $\gamma$ tend to 0
the query time must be $\Omega(N^{\frac{1}{2}-o(1)})$.
This CLB is comparable with the data structure of
Cohen and Porat~\cite{CP10} (see also~\cite{KPP15}) where $t_p=O(N\sqrt N)$ and $t_q = O(\sqrt N)$ .
Furthermore, if we only allow constant query time, then by making $\gamma $  tend to $1$
the preprocessing time must be $\Omega(N^{2-o(1)})$, matching that of the trivial preprocessing algorithm that
computes all answers in advance. A clean expression of this tradeoff curve is given in the following corollary; refer to Section~\ref{sec:set_disjoint_improved_proof} for the proof.

\begin{corollary}
\label{cor:set_disjoint_improved}
Assume the \ThreeSUM{} conjecture.
Fix constants $p\in [1,2)$ and $q\in [0,1/2]$
and suppose there is a data structure for \SetDisjointness{} where
$t_p = O(N^{p})$ and $t_q = O(N^{q})$.  Then
\[
p +2q \ge  2.
\]
\end{corollary}

\subsection{Triangle Enumeration}
\begin{figure*}
\centering\begin{tabular}{lll}
\multicolumn{3}{c}{\sc Triangle Enumeration Bounds}\\
Authors & Time Bound & Remarks \\\cline{1-3}
\hline
Itai \& Rodeh & $O(m^{3/2})$ & \\
\hline
Chiba \& Nishizeki & $O(m\alpha)$ & $\alpha = $ arboricity\\
\hline
\Patrascu & $\Omega(m^{4/3-o(1)})$ & $t \approx m = n^{1.5+o(1)}$ triangles\\
\hline
 & $O(m^{\frac{2\omega}{\omega+1}} + m^{\frac{3(\omega-1)}{\omega+1}}t^{\frac{3-\omega}{\omega+1}})$ &\istrut[2]{6}\\
\Bjorklund, Pagh, & $O(n^{\omega} + n^{\frac{3(\omega-1)}{5-\omega}}t^{\frac{2(3-\omega)}{5-\omega}})$ & \rb{3}{$\omega = $ matrix mult.~exp.}	\istrut[2]{6}\\\cline{2-3}
Williams \& Zwick & $O(m^{4/3+o(1)} + t\cdot (\frac{m}{t^{2/3}}))$ &\istrut[2]{5}\\
 &$O(n^{2+o(1)} + t\cdot (\frac{n}{t^{1/3}}))$ & \rb{3}{Assuming $\omega = 2$}\istrut[3]{5}\\\hline
Kopelowitz, Pettie \& Porat & $O(m\ceil{\frac{\alpha\log\log n}{\log n}} + t)$ & Randomized, w.h.p.\\\hline
Eppstein, Goodrich,		&									& Randomized, w.h.p.\\
Mitzenmacher \& Torres	& \rb{2.5}{$O(m\ceil{\frac{\alpha\log w}{w}} + t)$}	& $w = $ word size\\\hline
\hline
 & $\Omega(m\alpha^{1-o(1)})$ & every arboricity $\alpha$\\
\rb{2.5}{\bf new} & $\Omega(\min\{m^{3/2-o(1)}, \;t\cdot (\frac{m}{t^{2/3}})^{1-o(1)}\})$  & \\
\rb{2.5}{(Assuming \ThreeSUM{} Conjecture)} & $\Omega(\min\{n^{3-o(1)}, \;t\cdot (\frac{n}{t^{1/3}})^{1-o(1)}\})$  & \\ \cline{1-3}
\end{tabular}
\caption{\label{fig:triangle-bounds}
\Bjorklund{} et al.~\cite{BjorklundPWZ14} also proved several lower bounds conditioned on the
\QES{} Conjecture.  This conjecture was later refuted by
Lokshtanov, Patrui, Tamaki, Williams, and Yu~\cite{LokshtanovPTWY17}.}
\end{figure*}

Given a graph with $n$ vertices and $m$ edges, the \emph{Triangle Enumeration} problem is
to list all $t$ triangles (3-cycles), or alternatively, to list up to a given number $t$ of triangles.
The maximum number of triangles in a graph is $O(m^{3/2})$.  (For example, consider a clique.)
Itai and Rodeh~\cite{IR78} obtained a Triangle Enumeration algorithm running in $O(m^{3/2})$ time,
which is optimal inasmuch as it is linear in the worst case output size.  The Itai-Rodeh algorithm
can also run in $\Omega(m^{3/2})$ time even if there are zero triangles, which highlights the need for
more nuanced ways to measure the performance of enumeration algorithms.
Chiba and Nishizeki~\cite{CN85} showed that all triangles could be enumerated in
$O(m\alpha)$ time where $\alpha = \alpha(G)$ is the arboricity\footnote{The \emph{arboricity} of an undirected graph $G=(V,E)$ is
the number of forests needed to cover $E$; by the Nash-Williams theorem~\cite{Nash61,Nash64}
it is precisely $\alpha(G) =  \max_{U\subseteq V \,:\, |U|\ge 2} \left\lceil \frac{\card{E(U)}}{\card{U}-1} \right\rceil$,
where $E(U)$ is the set of edges induced by $U$.} of the graph.
The Chiba-Nishizeki bound subsumes the Itai-Rodeh bound because $\alpha$ is always upper bounded by $\sqrt{m}$;
it is also optimal inasmuch as the output size can be as large as $\Omega(m\alpha)$, for any $\alpha\ge 2$ and $m=O(n\alpha)$.

One of the big open questions in this area is to understand the relationship between
Triangle Detection and Triangle Enumeration, or more generally, to separate the one-time costs of enumeration
(e.g., in terms of $m,\alpha$)  from the per-triangle costs (in terms of $t$).
In particular, is $\Omega(m\alpha)$ time necessary for reasons other than worst case output size considerations?
Kopelowitz, Pettie, and Porat~\cite{KPP15}
proved that enumerating $t$ triangles takes $O(m\ceil{\alpha/\frac{\log n}{\log \log n}} + t)$ time.
Eppstein et al.~\cite{EppsteinGMT17} designed an algorithm for the $w$-bit word RAM model
running in $O(m\ceil{\alpha/\frac{w}{\log w}} + t)$ time.
Strictly speaking, these algorithms disprove the hypothesis that
$\Omega(m\alpha)$ is necessary, but perhaps it cannot be improved more than polylogarithmic factors.

\Patrascu{} showed that, conditioned on the \ThreeSUM{} conjecture, there exists a graph with
$t=O(m)$ triangles, for which listing all triangles must take $\Omega(m^{4/3-o(1)})$ time.
A careful examination of \Patrascu's proof shows that the arboricity of this graph is indeed roughly $m^{1/3}$,
implying the essential optimality of Chiba and Nishizeki's algorithm and of~\cite{KPP15,EppsteinGMT17},
at least for one particular arboricity.
However, \Patrascu's CLB does not extend to any $\alpha \gg m^{1/3}$.
This left open the possibility of
clever algorithms that dramatically improve the $\Omega(m\alpha)$ bound on dense graphs.
For example, in many problems, improvements based on fast matrix multiplication only ``kick in'' when the graph
is sufficiently dense.

Using our new framework we prove an $\Omega(m\alpha^{1-o(1)})$ CLB for Triangle Enumeration,
for all possible arboricities $1\ll \alpha \ll m^{1/2}$.
We emphasize that the number of triangles in these instances, $t$, is polynomially smaller than $m\alpha$,
implying that the hardness does not stem {\em solely} from the output size.
Thus, the Chiba-Nishizeki algorithm and the algorithms of~\cite{KPP15,EppsteinGMT17} are essentially optimal for the entire
spectrum of arboricities. The proof of Theorem~\ref{thm:triangle_lb} is given in Section~\ref{sec:triangle_lb}.

\begin{theorem}\label{thm:triangle_lb}
Assume the \ThreeSUM{} conjecture.
For any constants $x\in (0,1)$ and $y\in (0,1/2)$ such that $x\leq 2y$,
there exists a constant $\epsilon>0$ and a graph with $n$ vertices, $m$ edges,
arboricity $\alpha=\Theta(n^x)=\Theta(m^y)$,
and $t < m\alpha^{1-\epsilon}$ triangles,
such that listing all triangles requires $\Omega(m\alpha^{1-o(1)})$ expected time.
\end{theorem}

For a comparison between the known upper and conditional lower bounds, see Figure~\ref{fig:triangle-bounds}.

\paragraph{Output-sensitive triangle enumeration algorithms.} Another approach for enumerating triangles takes $t$ as a parameter
and enumerates up to $t$ triangles, even if the graph contains more.  The time for such algorithms has a one-time cost
(depending on graph parameters such as $n$ and $m$, but not $t$), and a cost \emph{per triangle}.
An algorithm of Bj{\o}rklund, Pagh, Williams, and Zwick~\cite{BjorklundPWZ14} shows that if the matrix multiplication exponent is
$\omega =2$, then listing up to $t$ triangles takes $\tilde{O}(\min\{n^2 + nt^{2/3}, \; m^{4/3} + mt^{1/3}\})$ time.
Notice that this runtime can be expressed as paying either a one-time cost of $n^2$ or $m^{4/3}$ and a per triangle cost of
$\frac{n}{t^{1/3}}$ or $\frac{m}{t^{2/3}}$.

We prove that, assuming the \ThreeSUM{} conjecture and assuming $\omega=2$, this per triangle cost is essentially optimal.
This lower bound is obtained by considering the extreme case of listing all triangles in the graph, which by Theorem~\ref{thm:triangle_lb} requires $\Omega(m\alpha^{1-o(1)})$ expected time, combined with controlling the number of triangles in the graph so that $t= \alpha^3$.
Such a result seems unobtainable in \Patrascu's framework since the corresponding graphs have arboricity at most $m^{1/3}$.
The following theorem is proved in Section~\ref{sec:triangle_lb}.

\begin{theorem}\label{thm:triangle_listing_lb}
Assume the \ThreeSUM{} conjecture.
Then any algorithm for listing $t$ triangles whose runtime is expressed in terms of the number of edges $m$ must take
$\Omega(\min\{m^{3/2-o(1)}, t\cdot (\frac{m}{t^{2/3}})^{1-o(1)}\})$ expected time.
If the runtime of the algorithm is expressed in terms of the number of vertices $n$ it must take $\Omega(\min\{n^{3-o(1)}, t\cdot (\frac{n}{t^{1/3}})^{1-o(1)}\})$ expected time.
\end{theorem}

In particular, Theorem~\ref{thm:triangle_listing_lb} implies that if we do not spend $\Omega(m^{3/2})$ time for listing just $t$ triangles (which is enough time to report all triangles), then the time per triangle must be $\Omega((\frac{m}{t^{2/3}})^{1-o(1)})$.

\subsection{Conditional Lower Bounds on Graph and Pattern Matching Problems}

We also prove polynomial CLBs for data structure versions of Document Retrieval problems,
Maximum Cardinality Matching (improving~\cite{AW14}),
 $d$-failure Connectivity Oracles, and Distance Oracles for Colors.
The new CLB for Maximum Cardinality Matching
is of particular interest since it introduces new techniques for
obtaining \emph{amortized} lower bounds; see Section~\ref{section:MCM}.

\paragraph{Maximum Cardinality Matching.}
In the \textit{Dynamic Maximum Cardinality Matching problem} we are interested in maintaining a dynamic graph $G=(V,E)$, with $n=|V|$ and
$m=|E|$, to support maximum cardinality matching (MCM) queries, which report the size of the current MCM.
When both insertions and deletions are supported we say that $G$ is fully dynamic, while if only insertions are supported we say that $G$ is \emph{incremental}. The trivial algorithm for updating an MCM takes $O(m)$ time by finding an augmenting path~\cite{GT91,GT85}.
Sankowski~\cite{Sankowski07} gave a fully dynamic algorithm with an amortized time bound of $O(n^{1.495})$
based on fast matrix multiplication.
In the bipartite vertex-addition model, where vertices on one side of the graph arrive online with all of their edges,
Bosek, Leniowski, Sankowski, and Zych~\cite{BLSZ14} recently showed how to maintain a maximum
cardinality matching whose total update time is $O(m\sqrt n)$ time; see also~\cite{BernsteinHR18}.

Abboud and Vassilevska Williams~\cite{AW14} showed that,
based on the \ThreeSUM{} conjecture, either the preprocesing time of Dynamic MCM is $\Omega(m^{4/3-o(1)})$,
the amortized update time is $\Omega(m^{\alpha-o(1)})$, or the amortized query time is $\Omega(m^{2/3-\alpha-o(1)})$,
for any $\alpha \in [1/6,1/3]$.
In our setting we will require the size of the MCM to be reported after each update, and so the CLB of~\cite{AW14} implies that if the preprocessing time $t_p$ is $O(m^{4/3-\Omega(1)})$ then the update time $t_u$ is $\Omega(m^{1/3-o(1)})$.
Using Theorem~\ref{thm:improved_reduction} we are able to prove the following.  Refer to Section~\ref{sec:MCM_prep}
for proof of Theorem~\ref{thm:dynamic-MCM-preprocessing}.

\begin{theorem}\label{thm:dynamic-MCM-preprocessing}
Assume the \ThreeSUM{} conjecture, and fix any $\gamma \in (0,1)$.
Suppose a fully dynamic MCM algorithm is given an MCM of the initial graph at preprocessing,
and let $t_p$ and $t_u$ be its preprocessing and update times, respectively.  Then
\[
t_p+m^{\frac{1+\gamma}{2-\gamma}}\cdot t_u = \Omega(m^{\frac 2{2-\gamma} - o(1)}).
\]
Moreover, the same bound holds even for the class of {\em approximate} MCM algorithms that report
the size of some matching without length-7 augmenting paths.
\end{theorem}

By having $\gamma$ approach $0$, one implication of Theorem~\ref{thm:dynamic-MCM-preprocessing} is that
if $t_p = O(m^{1+o(1)})$ then the update time must be $\Omega(m^{1/2-o(1)})$.
This improves on the results of Abboud and Vassilevska Williams~\cite{AW14}.\footnote{They~\cite{AW14} also get CLBs on approximate MCM algorithms that eliminate short augmenting
paths, but from different conjectures concerning the complexity of triangle detection and combinatorial BMM.}
Guaranteeing the absence of short augmenting paths is one way to achieve a provably good approximation to the MCM; see~\cite{NS13}.

\paragraph{Incremental Maximum Cardinality Matching.}
In this setting we consider an initially empty graph $G$ and so there is no preprocessing phase.
Abboud and Vassilevska Williams~\cite{AW14} show that their lower bounds for fully dynamic MCM extend to incremental MCM, but only for
{\em worst case} time bounds.  They highlight the difficulty of obtaining amortized lower bounds using their approach.
The worst-case lower bounds of Abboud and Vassilevska Williams~\cite{AW14} can be phrased in
terms of $\hat n$, the number of {\em vertices} when an operation takes place.
Either the update or query time is $\Omega({\hat n}^{1/2-o(1)})$.
It is straightforward to show that if the graph is allowed to grow with each query
then one can obtain an amortized expected $\Omega(\hat n^{1/3-o(1)})$ lower bound.

We focus on improving the amortized lower bound for incremental MCM in terms of $\hat n$,
using Theorem~\ref{thm:improved_reduction}.
Our strategy is to answer \SetDisjointness{} queries using an incremental dynamic MCM algorithm.
The construction has
the property that queries can be simulated by two vertex insertions and two edge insertions
and then examining the change in the size of the MCM.  There are two ways to undo these
four insertions, (i) rolling back the state of the data structure to its original state, or (ii) inserting two more vertices and two more edges.
By dynamically choosing which of (i) or (ii) to employ we can better control the total number of vertices inserted into the graph,
and therefore get better lower bounds in terms of $\hat{n}$.

\begin{theorem}\label{thm:MCM}
Assume the \ThreeSUM{} conjecture.
Any algorithm for incremental MCM has amortized expected update time of $\Omega(\hat{n}^{\frac{\sqrt {17} -1} 8 -o(1)}) = \Omega(\hat{n}^{0.3903-o(1)}) $
where $\hat{n}$ is the number of vertices in the graph following the update.
\end{theorem}

The proof of Theorem~\ref{thm:MCM} appears in Section~\ref{section:MCM}.
Dahlgaard~\cite{Dahlgaard16} later proved a conditional lower bound of
$\Omega(n^{1-o(1)})$  on incremental MCM, but from the \OMv{} conjecture.

\paragraph{Organization of the Paper.}
In Section~\ref{section:improved_reduction} we prove
Theorems~\ref{thm:improved_reduction} and \ref{thm:improved_reduction_reporting}
that reduce \ThreeSUM{} to \SetDisjointness{} and \SetIntersection.
In Section~\ref{sec:triangle_lb} we apply these results to the \emph{triangle enumeration} problem,
and prove Theorems~\ref{thm:triangle_lb} and \ref{thm:triangle_listing_lb}
concerning the optimality of the Chiba-Nishizeki~\cite{CN85} algorithm
and the \Bjorklund{} et al.~\cite{BjorklundPWZ14} algorithm, respectively.
Section~\ref{sect:MCM} presents proofs of Theorems~\ref{thm:dynamic-MCM-preprocessing}
and \ref{thm:MCM}
on fully dynamic cardinality matching
and incremental cardinality matching, respectively.
In Section~\ref{sec:applications} we provide a number of new conditional lower bounds
for \emph{online} \SetIntersection{} data structures, $d$-failure connectivity oracles, and
various document retrieval problems.  Finally, in Section~\ref{app:proof_conv3sum}, we prove
Theorem~\ref{thm:conv3sum} on the near-equivalence between
the randomized complexities of \ThreeSUM{} and \ConvolutionThreeSUM.

\section{The Main Reductions --- Proofs of Theorems~\ref{thm:improved_reduction} and \ref{thm:improved_reduction_reporting}}\label{section:improved_reduction}

Let $\mathcal{H}$ be a family of hash functions from $[u] \rightarrow [m]$.
We call $\mathcal{H}$ \emph{linear} if
\begin{align*}
&\mbox{for any $h\in\mathcal{H}$ and any $x,x' \in [u]$, we have $h(x) + h(x') \equiv h(x+x') \; (\modulo m)$.}\\
\intertext{and call $\mathcal{H}$ \emph{balanced} if}
&\mbox{for any $h\in \mathcal{H}$ and any $S\subset [u]$, $\card{\{x\in S : h(x) = i\}} \leq \frac{3|S|}{m}$.}
\end{align*}
Note that \emph{no} hash family is balanced according to this definition (unless $m=O(1)$).  Nonetheless,
in the first part of the proof we will engage in a little magical thinking, and suppose that for any $u$ and $m$
a \emph{hash fairy} can summon for us a hash family $\mathcal{H}$ that is magically linear, balanced, and pairwise independent.
In Section~\ref{sect:almost_linear_and_balanced} we argue that a modified version of the proof works
even if $\mathcal{H}$ is \emph{almost} linear, \emph{almost} balanced, and (exactly) pairwise independent,
and exhibit a specific hash function that satisfies these properties.

\begin{proof}[Combined proof of Theorem~\ref{thm:improved_reduction} and Theorem~\ref{thm:improved_reduction_reporting}]
Since the proofs of both theorems follow a similar path we describe them together.
We are looking for three distinct input elements $x,y,z$ such that $x-y=z$.
Let $R=n^{\gamma}$
and set $Q = (5n/R)^2$ in Theorem~\ref{thm:improved_reduction} and
$Q = n^{1+\delta}/R$ in Theorem~\ref{thm:improved_reduction_reporting}.
Without loss of generality we assume that $\sqrt Q$ is an integer.
Finally, we assume that the input is drawn from the integer universe $[2^w]$,
where $w=\Omega(\log n)$ is the machine's word length.

We pick a random hash function $h_1 : [2^w] \rightarrow [R]$ from a family that is linear and balanced.
Using $h$, we create $R$ buckets $\mathcal{B}_1,\ldots, \mathcal{B}_{R}$
such that $\mathcal{B}_i=\{x:h_1(x)=i\}$.
Since $h_1$ is balanced, each bucket contains at most $3n/R$ elements.
This bucketing is similar to \Patrascu's reduction~\cite{Patrascu10}.

Next, we pick a random hash function $h_2: [2^w] \rightarrow [Q]$ where
$h_2$ is chosen from a pair-wise independent and linear family.
For each bucket $\mathcal{B}_i$ we create $2\sqrt{Q}$ shifted sets as follows.
For each $j \in [0,\sqrt{Q})$
let
\begin{align*}
			&\mathcal{B}_{i,j}^{\uparrow} 	= \{(h_2(x)+ j\cdot \sqrt{Q}) \, \modulo Q \;|\; x\in \mathcal{B}_i\}\\
\mbox{ and \ } &\mathcal{B}_{i,j}^{\downarrow} 	= \{(h_2(x)- j) \, \modulo Q \;|\; x\in \mathcal{B}_i\}.
%
\intertext{Next, for each $z\in A$ we want to determine if there exists $x$ and $y$ in $A$ such that $x-y=z$. To do this we utilize the linearity of $h_1$ and $h_2$,
which implies that}
%
			&h_1(x) - h_1(y) \equiv h_1(z) \, (\modulo R)\\
\mbox{ and \ } 	&h_2(x) - h_2(y) \equiv h_2(z) \, (\modulo Q).
\end{align*}
If $x\in \mathcal{B}_i$ then $y$ must be in bucket $\mathcal{B}_{i-h_1(z) \,\modulo R}$.
Thus, for each $i\in [R]$ we would like check the intersection
$\mathcal{B}_i \cap (\mathcal{B}_{i-h_1(z) \,\modulo R}+z)$
to find candidate pairs $x,y$ for which $x-y=z$.
Denote the high-order and low-order halves of $h_2$ by
\begin{align*}
			h_2^{\uparrow}(z) &= \floor{\frac{h_2(z)}{\sqrt{Q}}}\\
\mbox{ and \ } h_2^{\downarrow}(z) &= h_2(z) \, \modulo \sqrt{Q}.
\end{align*}
Due to the linearity of $h_2$, every element in the intersection of
\begin{align}
& \mathcal{B}_i \cap (\mathcal{B}_{i-h_1(z) \,\modulo R}+z) 		\label{eqn:intersect}
\intertext{has a corresponding element in the intersection of}
&\paren{\mathcal{B}_{i,h_2^{\downarrow}(z)}^{\downarrow}} \cap \paren{\mathcal{B}_{i-h_1(z) \,\modulo R,h_2^{\uparrow}(z)}^{\uparrow}}.  \label{eqn:h2intersect}
\end{align}
Of course, the reverse direction is not true since taking the projection of these sets under $h_2$
may introduce false positives into (\ref{eqn:h2intersect}) that were not present in (\ref{eqn:intersect}).
Nonetheless, if (\ref{eqn:h2intersect}) is empty then (\ref{eqn:intersect}) is empty as well, meaning
there are no \ThreeSUM{} witnesses involving $z$ and any $x\in \mathcal{B}_i$.
The number of set intersection queries is $nR$ since there are $n$ choices for $z$ and $R$ choices for $i$.

Fix $z$ and let $k = h_2(z)$.
Since $h_2$ is pairwise independent and linear then for any pair $x,y\in U$ where $x\neq y$ we have that if $x-y\neq z$ then
\[
\Pr[h_2(x) - h_2(y) = k] = \Pr[h_2(x-y) = h_2(z)] =\frac{1}{Q}.
\]
This is where the proofs of the two theorems diverge.

\paragraph{Details for Theorem~\ref{thm:improved_reduction}.}
Since each bucket contains at most $3n/R$ elements, the probability of a false positive stemming from
two buckets $\mathcal{B}_i,\mathcal{B}_j$ and a given offset $k=h_2(z)$ is, by a union bound,
\[
\Pr[h_2(\mathcal{B}_i) \cap (h_2(\mathcal{B}_j)+k) \neq \emptyset] \leq \paren{\frac{3n}{R}}^2 \frac 1 {Q} = \frac 9 {25}.
\]
(Recall that $Q=(5n/R)^2$.)
In order to reduce the probability of false positives, we repeat the process with $O(\log n)$
different choices of $h_2$, but using the same $h_1$.
This blows up the number of sets by a factor of $O(\log n)$, but not the universe.
If the sets intersect under all $O(\log n)$ choices of $h_2$ then we spend
$O(n/R)$ time to find $x$ and $y$ within buckets $\mathcal{B}_i$ and $\mathcal{B}_j$,
which is either part of a \ThreeSUM{} witness (and the algorithm halts),
or a false positive, which only occurs with probability $1/\mbox{poly}(n)$.

\paragraph{Details for Theorem~\ref{thm:improved_reduction_reporting}.}
We bound the expected number of false positives.
Each pair of buckets defines at most $(\frac {3n} R)^2$ pairs of elements, one from each bucket,
so the expected number of false positives arising from this intersection is
\[
E[|h_2(\mathcal{B}_i) \cap (h_2(\mathcal{B}_j)+k)|] = \paren{\frac {3n} R}^2 \frac 1 {Q} = O\paren{\frac{n^{1-\delta}}{R}},
\]
since $Q = \Theta(n^{1+\delta}/R)$.
Thus, the expected number of false positives over all
$O(nR)$ intersections is
$O(nR\frac{n}{Rn^\delta}) = O(n^{2-\delta}) \le c\cdot n^{2-\delta}$ from some constant $c>0$.
It takes constant time to verify that a pair $x,y$ is a false positive rather than (part of) a valid \ThreeSUM{} witness.
If the number of verifications exceeds $2c\cdot n^{2-\delta}$ without finding a \ThreeSUM{} witness,
which, by Markov's inequality, happens with probability at most $1/2$, then we restart the entire algorithm with
fresh hash functions.

\paragraph{Final details.}
To summarize, for \SetDisjointness{} (\SetIntersection{}) we create a total of $O(R\sqrt{Q}\log n)$ sets ($O(R\sqrt{Q})$ sets). These sets are partitioned into two families $A$ and $B$ where all of the $\uparrow$-type sets are in $A$ and all of the $\downarrow$-type sets are in $B$. All of the intersections we are interested in are between a set from $A$ and a set from $B$. The universe $U$ of the elements in the sets is of size $Q$. The family $\mathcal{F}$ is $A\cup B$.
The number of queries is $O(nR\log n) = O(n^{1+\gamma}\log n) $  ($O(nR) = O(n^{1+\gamma}) $).

This concludes the proof of Theorems~\ref{thm:improved_reduction} and \ref{thm:improved_reduction_reporting},
under the simplifying assumption that $h_1,h_2$ are magically linear and balanced, which is addressed
in Section~\ref{sect:almost_linear_and_balanced}.
\end{proof}

\subsection{Almost Linear and Almost Balanced Hashing}\label{sect:almost_linear_and_balanced}

We now describe how to overcome the assumption that there exist pair-wise independent hash functions that are
magically linear and balanced by relaxing both definitions.
A family $\mathcal{H}$ of hash functions from $[u] \rightarrow [m]$
is called {\em almost linear} if for any $h\in\mathcal{H}$ there exists an integer
$c_h$ such that for any $x,x' \in [u]$,
\begin{align*}
h(x) + h(x') 	&\equiv h(x+x')  + c_h + \{-1,0,1\} \: (\modulo m)
\end{align*}
I.e., it is linear, up to a offset $c_h$ that depends on $h$ and a $\pm 1$ error.
Given a hash function $h\in \mathcal{H}$ we say that a value $i\in m$ is \emph{heavy} for set
$S \subseteq [u]$ if
\[
\card{\{x\in S : h(x) = i\}} > \frac{3|S|}m,
\]
i.e., 3 times more than the expected load.
$\mathcal{H}$ is called {\em almost balanced} if for any set $S\subseteq [u]$,
the expected number of elements from $S$ that are hashed to heavy values is $O(m)$.

We emphasize that the notion of almost linearity that we use here is more general than the one used
by Baran et al.~\cite{BaranDP08}, \Patrascu~\cite{Patrascu10}, and Section~\ref{app:proof_conv3sum} of this paper.
Whereas Section~\ref{app:proof_conv3sum} can use any $O(1)$-universal
almost linear hash family, here we \emph{require} a pairwise independent (and hence exactly 1-universal) almost linear hash family.
The hash family~\cite{DietzfelbingerHKP97}
proposed by Baran et al.~\cite{BaranDP08} and \Patrascu~\cite{Patrascu10} is almost linear (under Section~\ref{app:proof_conv3sum}'s definition)
but it is only known to be 2-universal~\cite[Lem.~2.4]{DietzfelbingerHKP97}
and is definitely not pairwise independent.  This issue was also noted in~\cite{JafargholiV13}.

We will show that there exists a family $\mathcal{H}$ of pairwise independent hash functions that is almost linear and almost balanced,
which is suitable for use in the reductions of Theorems~\ref{thm:improved_reduction} and~\ref{thm:improved_reduction_reporting}.
Each step in the proofs of Theorems~\ref{thm:improved_reduction} and~\ref{thm:improved_reduction_reporting}
that used linearity is replaced by three parallel steps making use of the almost linearity of $\mathcal{H}$.
The reduction algorithm must consider all three options for $\zeta \in \{-1,0,1\}$ such that
$h(x) + h(x')\equiv h(x+x') + c_h + \zeta\; (\modulo m)$.
This only blows up the running time by a factor of nine: whereas before we assumed $h_1$ and $h_2$ were perfectly linear,
we now have to entertain three options for $h_1$ hash values and three options for $h_2$ hash values.

The balance assumption is overcome by directly verifying, for each element mapped to a heavy value,
whether it is part of a \ThreeSUM{} witness. This takes $O(n)$ time per element.
The expected number of such elements is $O(n^\gamma)$, and so the expected time spent on such elements is $O(n^{1+\gamma})$.
The number of intersection queries $q$ is already $O(n^{1+\gamma})$
and so the cost of dealing with elements mapped to heavy values may be ignored.
For the rest of the elements (those that are not assigned to heavy values) we
proceed as in the proofs of Theorem~\ref{thm:improved_reduction} and~\ref{thm:improved_reduction_reporting}.

\paragraph{The Hash Family.}
Baran et al.~\cite{BaranDP08} showed that any 1-universal family of hash
functions is almost balanced; see Jafargholi and Viola~\cite{JafargholiV13} for a somewhat simpler proof.\footnote{Furthermore, it
is straightforward to exhibit a class of $(1+\epsilon)$-universal functions that are not close to being almost balanced;
see, e.g., https://simons.berkeley.edu/talks/seth-pettie-2015-11-30, starting at minute 16:00.}
Rather than use the family of~\cite{DietzfelbingerHKP97}, we use one analyzed by Dietzfelbinger~\cite{Dietzfelbinger96}.
\begin{theorem} (\cite[Theorem 3]{Dietzfelbinger96})
The family $\mathcal{H}_{u,m,r}$ defined below is pairwise independent and hence 1-universal
whenever $r=km$ for some $k\geq u/2$, and $u,m,$ and $r$ are all powers of 2.
\begin{align*}
\mathcal{H}_{u,m,r} & =
\left\{
h_{a,b} : [u]\rightarrow[m] \;\,| \mbox{ $a\in [r]$ is an odd integer and $b\in [r]$}
\right\}\\
\mbox{and} \; h_{a,b}(x) & = ((ax + b) \div (r/m)) \, \modulo m
\end{align*}
\end{theorem}

In our application $u=2^w$ is naturally a power of 2, $m$ can be rounded up to the next power of 2,
and $k$ can be fixed at $u/2$.
Since $\mathcal{H}_{u,m,r}$ is 1-universal it is also almost balanced~\cite{BaranDP08,JafargholiV13}.
We need to prove that it is almost linear.

\begin{lemma}
The family $\mathcal{H}_{u,m,r}$ is almost linear, with $c_{h_{a,b}} = b \div (r/m)$.
\end{lemma}
\begin{proof}
Consider any elements $x,x'\in [u]$.  Let $g(x) = (ax+b)\div(r/m)$ be the hash function without the ``$\modulo m$'' operation.
Then we have
\begin{align*}
g(x) + g(x') 		&= \floor{\frac{ax + b}{r/m}} + \floor{\frac{ax' + b}{r/m}}\\
g(x+x') + b \div (r/m)	&= \floor{\frac{a(x+x')+b}{r/m}} + \floor{\frac{b}{r/m}}
\end{align*}
In general, whenever $\alpha_1+\alpha_2 = \alpha_3+\alpha_4$, $\floor{\alpha_1} + \floor{\alpha_2}$ differs
from $\floor{\alpha_3} + \floor{\alpha_4}$ by at most one.
Hence, $g$ is almost linear with offset $b \div (r/m)$ and error in $\{-1,0,1\}$.
Taking $g$ modulo $m$ preserves almost linearity ($\modulo m$),
hence $h$ is also almost linear.
\end{proof}

\section{Triangle Enumeration}\label{sec:triangle_lb}

Following \Patrascu~\cite{Patrascu10} we express a \SetIntersection{} instance as a tripartite graph in which
triangles are in one-to-one correspondence with the elements output by \SetIntersection{} queries.

\begin{proof}[Proof of Theorem~\ref{thm:triangle_lb}]
The \SetIntersection{} instance of
Theorem~\ref{thm:improved_reduction_reporting} is interpreted as a tripartite graph $G$ on vertex set $A\cup B \cup U$ where $A$ and $B$ are two copies of $\mathcal{F}$.
Each element $e\in U$ has edges to the sets in $A$ and $B$ that contain $e$,
and the edges between $A$ and $B$ correspond to the \SetIntersection{} queries.
Thus, $|U|=\Theta(n^{1+\delta-\gamma}) $, $|A|=|B|=|\mathcal{F}|=\Theta(\sqrt{n^{1+\delta+\gamma}})$, there are $\Theta(n^{1+\gamma}) $ edges between $A$ and $B$, and at most $O(n^{1-\gamma}) $ edges between each
vertex in $A\cup B$ and elements in $U$. Thus the total number of edges between $A\cup B$ and $U$ is at most
$O(n^{\frac{1}{2}(1+\delta+\gamma)} \cdot n^{1-\gamma})
=
O(n^{\frac{1}{2}(3+\delta-\gamma)})$.
Notice that there is a bijection between the output elements in the
\SetIntersection{} instance and the triangles in the graph. Thus, after enumerating $O(n^{2-\delta})$ triangles we are done.

We prove that enumerating all triangles essentially requires $\Omega(M\alpha)$ time for any feasible combination of $N$ (the number of vertices), $M$ (the number of edges), and arboricity $\alpha$.
By feasible we mean that
$\alpha = \Theta(M^x) = \Theta(N^y)$, $y\leq 2x$.
We assume that each vertex has a degree of at least 2 (since it is easy to filter out other vertices) and then $x\leq y$. Notice that proving a lower bound for $\alpha=\Theta(N^x)$ implies a lower bound for $\Theta(N^{x'})$ for any constant $x'<x$ since one can always add singleton vertices.

For our lower bound proof it suffices to consider the case where
$n^{1+\gamma}  \geq n^{\frac{1}{2}(3+\delta-\gamma)}$, and so $\gamma \geq 1/3 + \delta/3$.
Furthermore, this implies that $|\mathcal{F}| > |U|$.
Thus,
\begin{align*}
N &= \Theta(n^{\frac{1}{2}(1+\delta+\gamma)})\\
M &= \Theta(n^{1+\gamma}).
\end{align*}
Our aim is to show that $\alpha$ is at most $O(n^{1-\gamma }) $, since for each edge we must spend at least $\Omega(n^{1-\gamma}) $ time (assuming the \ThreeSUM{} conjecture), and so if $\alpha \leq O(n^{1-\gamma}) $ we conclude that the total runtime is at least $\Omega(M\alpha)$. However, this may not be the case in $G$, so we devise a \emph{triangle-preserving}
reduction to a new graph $G'$ with $N'$ vertices and $M'=O(M)$ edges such that there is an injective function between triangles in the original graph and triangles in $G'$. To bound the arboricity $\alpha'$ of $G'$ we show that there exists an orientation of $G'$ with max out-degree
$O(n^{1-\gamma} )$. It is well known that the maximum out-degree in any orientation must be at least $\alpha-1$~\cite{KKPS14}.

Denote by $E(u,X)$ the set of edges between a vertex $u$ and a vertex set $X$.
Consider a vertex $a \in A$. Since $|E(a,U)| = O(n^{1-\neweps})$, if $|E(a,B)| = O(n^{1-\neweps})$ then we orient all of the edges of $a$ to leave $a$.
However, it is possible that $E(a,B)$ is too large. To deal with this, we create $\lceil\frac{|E(a,B)|}{n^{1-\gamma}}\rceil$ copies of $a$.
The neighbors of $a$ in $B$ are arbitrarily partitioned into $\lceil\frac{|E(a,B)|}{n^{1-\gamma}}\rceil$ sets of size at most $n^{1-\gamma}$,
and the $i$th copy of $a$ has as its neighbors the $i$th set in the partition. All of the edges touching copies of $a$ are oriented outwards from those copies.
Furthermore, each copy of $a$ has outgoing edges towards the $O(n^{1-\neweps})$ neighbors of $a$ in $U$.
Thus, the out-degree of each copy of $a$ is also at most $O(n^{1-\gamma})$. By orienting all of the edges between $B$ and $U$ to leave $B$, the out-degree of any vertex in this orientation is at most $O(n^{1-\gamma} )$, and so the arboricity of this new graph is at most $\alpha'=O(n^{1-\gamma})$.
It is straightforward to see that this new graph $G'$ is a triangle-preserving substitute for $G$.
The number of edges of this graph is $M' \le 2M$ since we only increase the number of edges by adding new edges between copies and $U$, but each such edge can be charged to an edge between $A$ and $B$ in the initial graph. Also, since there are $\Theta(n^{1+\gamma})$
edges between $A$ and $B$, the number of copies that are created is at most
$O(n^{1+\gamma} /n^{1-\gamma})=O(n^{2\gamma})$.
Hence, the number of vertices in the new graph is
$N'=O(N+n^{2\gamma}) = O(n^{\fr{1}{2}(1+\delta+\gamma)} + n^{2\gamma})$.
Finally, since $\gamma \geq 1/3 + \delta/3$
we have $(1+\delta+\gamma)/2 \leq 2\gamma$ and so $N' =O(n^{2\gamma})$.

To summarize, we have obtained a graph with
$M'=\Theta(n^{1+\gamma} )$ edges,
$N'=O(n^{2\gamma})$ vertices, and $\alpha' = O(n^{1-\gamma})$.
Thus, enumerating all of the triangles in
$O(M'(\alpha')^{1-\epsilon}) = O(n^{2-\Omega(\epsilon)})$
time contradicts the \ThreeSUM{} conjecture.
We will now show that this lower bound holds for the entire
spectrum of possible polynomial dependencies of $\alpha'$ on $N'$ and $M'$.

Recall that we always have $M'\leq N'\alpha'$. Since we can always increase the number of vertices, it is enough to prove that the lower bound holds for all combinations of $M'=N'\alpha'$.
This is exactly the case here, since
\[
M'=\Theta(n^{1+\gamma}) = \Theta(n^{2\gamma}n^{1-\gamma})  = \Theta(N'\alpha').
\]
Furthermore, we capture the entire spectrum of values of $\alpha'$.
To see this for $M'$ notice that $\alpha'$ is on the order of $M'^{\frac{1-\gamma}{1+\gamma}} = M'^x$.
As $\gamma$ admits values between $1/3$ and $1$ (exclusive), $x$ admits values between $1/2$ and $0$ (exclusive).
Similarly, $\alpha'$ is on the order of $N'^{\frac{1-\gamma}{2\gamma}} = N'^y$, so $y$ admits values between $0$ and $1$ (exclusive).
In this case the number of triangles, $n^{2-\delta}$, is bounded away from the lower bound
$\Omega(n^{2-o(1)})$ whenever $\delta > 0$, which illustrates that the
lower bound is not \emph{solely} a result of the size of the output.
\end{proof}

Using the same construction we can also prove the Theorem~\ref{thm:triangle_listing_lb},
which is in terms of $m,n,$ and $t$ rather than $m$ and $\alpha$.

\begin{proof}[Proof of Theorem~\ref{thm:triangle_listing_lb}]
Suppose, for the sake of obtaining a contradiction, that there exists an algorithm for enumerating $t$
triangles that takes $M^{(3/2)(1-\rho)} + t(M/t^{2/3})^{1-\rho}$ time, for some constant $\rho>0$.
We use the same graph construction as in Theorem~\ref{thm:triangle_lb},
setting $\gamma = 1/3 + \delta/3$ and $\delta = 4\rho$.
Thus, the graph has $N=\Theta(n^{2\gamma})$ vertices, $M = \Theta(n^{1+\gamma})$ edges,
and $t= \Omega(n^{2-\delta}) = \Omega(n^{3-3\gamma})$ triangles.
The running time of the algorithm on this instance
is therefore
\begin{align*}
	& M^{(3/2)(1-\rho)} + t(M/t^{2/3})^{1-\rho}\\
	&= (n^{1+\gamma})^{(3/2)(1-\rho)} + n^{3-3\gamma}(n^{1+\gamma - (2/3)(3-3\gamma)})^{1-\rho}\\
	&= n^{(4/3)(1+\delta/4)(3/2)(1-\rho)} +  n^{2 - \Theta(\rho)}		& \gamma = 1/3 + \delta/3\\
	&= n^{2 - \Theta(\rho^2)} + n^{2 - \Theta(\rho)}		& \rho=\delta/4.
\intertext{This contradicts the \ThreeSUM{} conjecture.   The calculations that depend on $N$ rather than $M$ are similar.
Suppose there is an algorithm that enumerates $t$ triangles in $N^{3(1-\rho)} + t(N/t^{1/3})^{1-\rho}$ time.  Then,
using the same graph construction, but with $\delta=\rho$, we can solve \ThreeSUM{} in time}
	&    N^{3(1-\rho)} + t(N/t^{1/3})^{1-\rho}\\
	&= n^{6\gamma(1-\rho)} + n^{3-3\gamma}(n^{2\gamma - (1/3)(3-3\gamma)})^{1-\rho}\\
	&= n^{2(1+\delta)(1-\rho)} + n^{2 - \Theta(\rho)}				& \gamma = 1/3 + \delta/3\\
	&= n^{2 - \Theta(\rho^2)} + n^{2 - \Theta(\rho)}				& \delta = \rho,
\end{align*}
which contradicts the \ThreeSUM{} conjecture.
\end{proof}

\section{Maximum Cardinality Matching}\label{sect:MCM}

\subsection{Incremental MCM --- Proof of Theorem~\ref{thm:MCM}}\label{section:MCM}

In this section $n$ denotes the size of the \ThreeSUM{} instance and
$N$ and $M$ denote the number of vertices and edges in the graph
on which we compute maximum cardinality matchings.

\begin{figure*}
\centering
\begin{tabular}{ccc}
\scalebox{.3}{\includegraphics{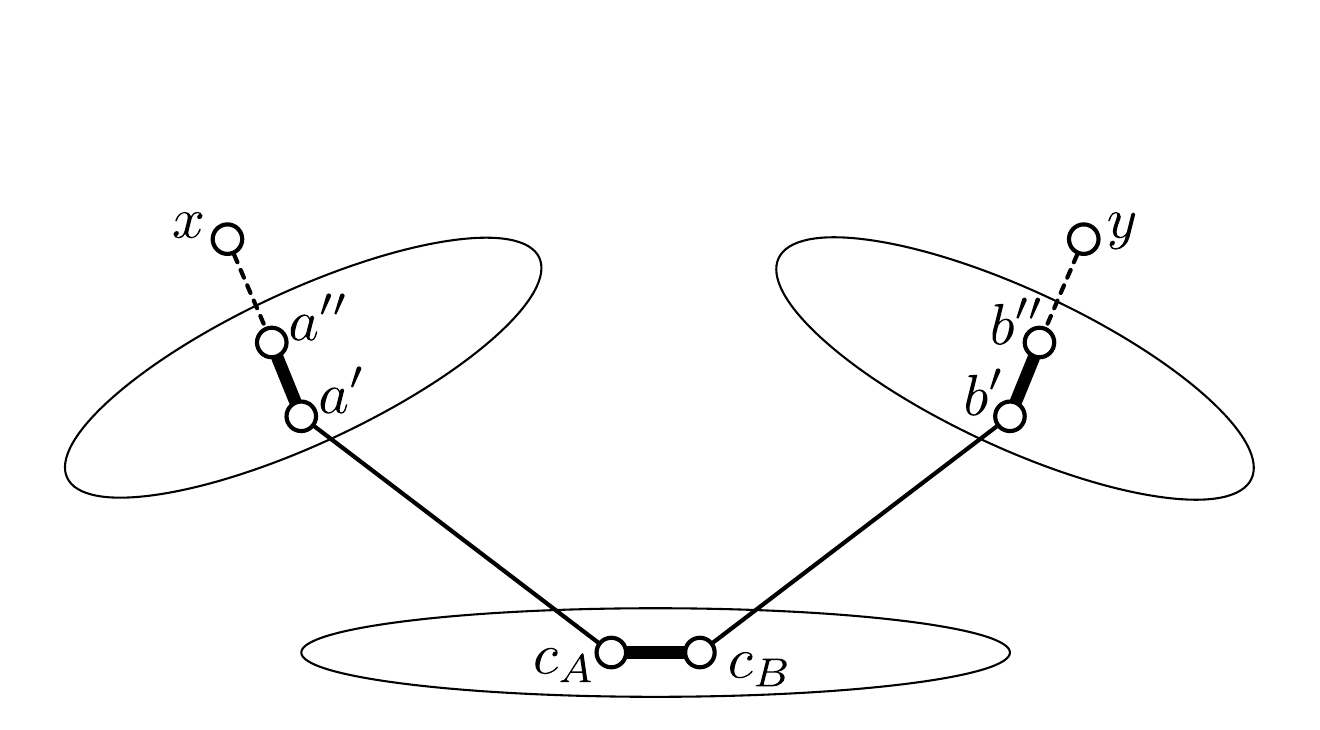}}
&  \hspace*{0cm}
&  \scalebox{.3}{\includegraphics{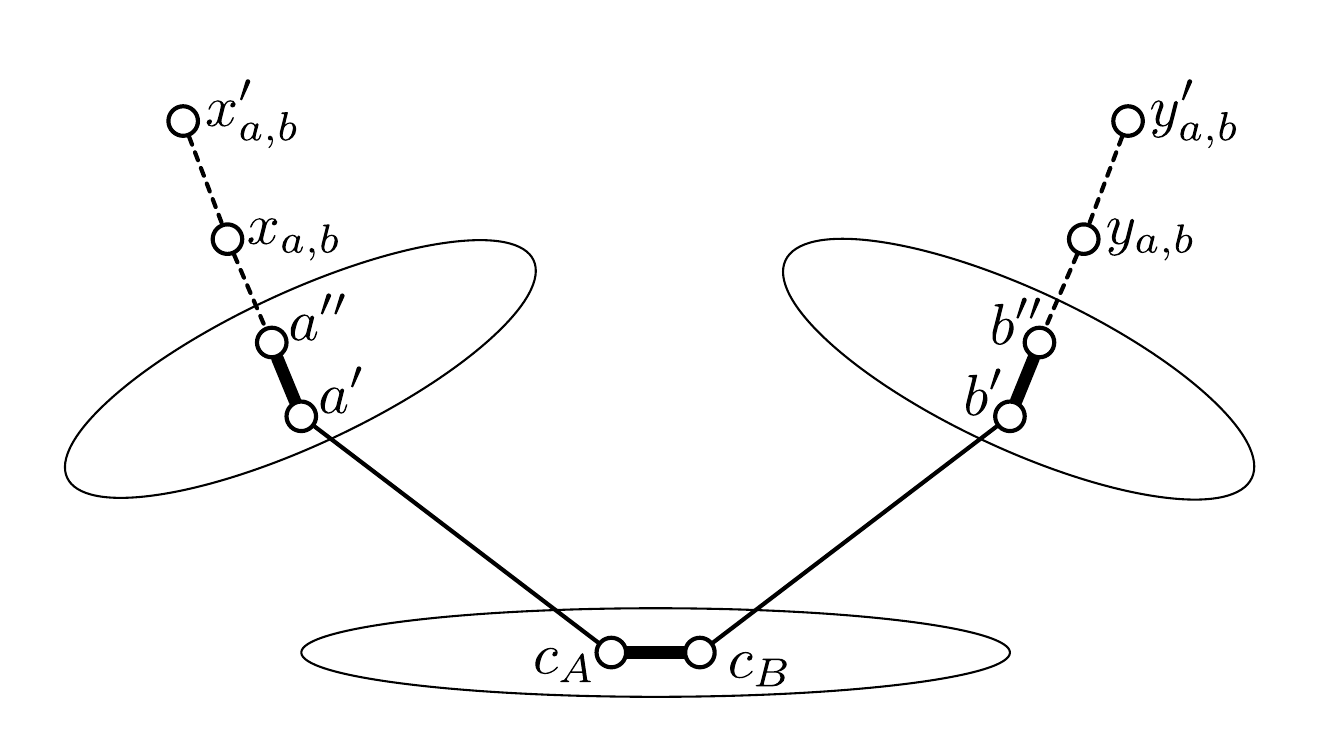}}\\
(A) && (B)
\end{tabular}
\caption{\label{fig:MCM} An illustration of the graph before the set
intersection query
$a\cap b = \emptyset$?.  There is a unique perfect matching before the
query: matched edges are drawn thick and unmatched ones thin.  Dashed
edges are inserted in the course of the query.  (A) For the worst case
bound we insert edges $(x,a''),(y,b'')$, check if the size of the MCM
has increased (implying $a\cap b\neq \emptyset$), then delete them.
(B) For the amortized bound we insert new vertices
$x_{a,b},x_{a,b}',y_{a,b},y_{a,b}'$ insert edges
$(x_{a,b},a''),(y_{a,b},b'')$, then check whether the size of the MCM
has increased,
then insert edges $(x_{a,b}',x_{a,b}),(y_{a,b}',y_{a,b})$.  Depending
on the actual
time of these operations, we either do nothing or
roll back all edge and vertex insertions.}
\end{figure*}


In amortized analysis we want to bound the total cost of a sequence $S=(\sigma_1,\ldots,\sigma_k)$ of $k$ operations and let $cost(\sigma_i)$ be the time cost of operation $\sigma_i$.
A function $f$ assigns valid \emph{amortized} costs if $\sum_{i=1}^k f(\sigma_i) \ge \sum_{i=1}^k \mbox{cost}(\sigma_i)$.
We prove that in any MCM algorithm, if $N_i$ is the number of vertices after
$\sigma_i$, then $\sum_{i=1}^k \mbox{cost}(\sigma_i)\geq \Omega\paren{\sum_{i=1}^k N_i^{\frac{\sqrt {17} -1} 8 -o(1)}}$,
that is, any amortization function $f$ that is a function of the current number of vertices $\hat{N}$ has
$f(\hat{N}) = \Omega\paren{\hat{N}^{\frac{\sqrt {17} -1} 8-o(1)}} = \Omega(\hat{N}^{0.39})$.

Consider the following instance of the incremental MCM problem which is created from an instance of the offline \SetDisjointness{} problem. First, we create two copies of $\mathcal{F}$ which are denoted by $A$ and $B$. For each $c\in U$ we create two vertices $c_A$ and $c_B$ with an edge between them. We say that $c_A$ and $c_B$ are copies of $c$. For each $a\in A$ ($b\in B$) we create two vertices $a'$ and $a''$ ($b'$ and $b''$) with an edge between them, and for each $c\in a$ ($c\in b$) there is an edge between $a'$ and $c_A$ ($b'$ and $c_B)$. We say that $a'$ and $a''$ ($b'$ and $b''$) are copies of $a$ ($b$).
We also add 2 additional vertices, $x$ and $y$.

The initialization of this graph is implemented by inserting all of the edges one at a time using the incremental MCM algorithm. This initial graph has $\Theta(n+n^{2-2\gamma})$ vertices and $\Theta(n^{2-\gamma})$ edges.
Notice that this initial graph (without the extra 2 vertices $x$ and $y$) has a unique perfect matching which is the set of edges between copies.
To implement a \SetDisjointness{} query between $a$ and $b$ we add edges $(x,a'')$ and $(y,b'')$. See Figure~\ref{fig:MCM}(A).
Now, $a$ and $b$ are disjoint iff there is no augmenting path after adding the two edges.
Thus, an increase in the MCM implies that $a\cap b \neq \emptyset$.
In order to facilitate additional \SetDisjointness{} queries we undo the {\em effect} of adding $(x,a''),(y,b'')$, using one of the following two approaches.

\paragraph{Rollback.} One approach is to delete the two edges that were added. An incremental data structure can always support
deletions of the last element that was inserted by keeping track of the memory modifications that took place during the last insertion, and reversing them
within the same time cost as the cost of the insertion itself.  Notice that this approach does not blend well with amortized time bounds since we have no a priori bound on the maximum time per operation.

\paragraph{Creating Perfect Matchings.} The second approach is to add another two edges per \SetDisjointness{} query for a total of 4 edges,
and create four separate dummy vertices $x_{a,b}, x'_{a,b},y_{b,a},$ and $y'_{b,a}$ associated with each \SetDisjointness{} query on $a\in A$ and $b\in B$.
For the \SetDisjointness{} query we add edges $(x_{a,b},a'')$ and $(y_{b,a},b'')$ to the graph and as before the MCM increases iff the sets $a$ and $b$ intersect. See Figure~\ref{fig:MCM}(B). Then, we add edges $(x_{a,b},x'_{a,b})$ and $(y_{b,a},y'_{b,a})$ which guarantee that the resulting graph has a perfect matching that is comprised of the perfect matching of the graph prior to the insertion of the 4 edges together with edges $(x_{a,b},x'_{a,b})$ and $(y_{b,a},y'_{b,a})$. The MCM has increased by 2 after the insertion of the 4 edges, regardless of  whether the sets intersect or not. The downside of this approach is that the number of vertices grows with the number of \SetDisjointness{} queries, leading to weaker lower bounds in terms of the number of vertices.

\paragraph{Combining the Two.}
Assume that the amortized cost of each edge or vertex insertion is
bounded by $\hat{N}^{\alpha}$ for some constant $\alpha >0$, where $\hat N$ is the number of vertices in the graph when the insertion takes place.
To answer a \SetDisjointness{} query we first add the four vertices and four edges, thereby creating a perfect matching. If the insertion time of these vertices and edges is less than $9 \hat N^\alpha$ (within $\hat{N}^{\alpha}$ of the budget for 8 insertions) then we rollback the insertion.
Otherwise, we leave the four edges and continue to the next \SetDisjointness{} query.
Intuitively, our goal with this combined method is to guarantee that the graph does not
grow by too much, while maintaining the amortized cost in order to obtain a higher lower bound.
\SetDisjointness{} queries for which we perform a rollback cost $O(\hat{N}^\alpha)$ time each.
By assumption the subsequence of remaining operations (those inserts used to set up the initial graph and subsequent inserts not rolled back)
 has amortized cost $O(\hat{N}^\alpha)$.

Next we bound the number of vertices at the end of the process, denoted by $N$. After the graph setup there is $O(n^{2-\gamma}(n+n^{2-2\gamma})^\alpha)$ credit for performing expensive insertions later. For our proof we will focus on $\gamma \leq 1/2$ and so the amount of credit becomes $O(n^{2-\gamma +(2-2\gamma)\alpha})$.
Each expensive \SetDisjointness{} query uses up at least $\hat{N}^\alpha$ of that credit, and so the total credit used during all of the expensive insertions is at least $\Omega(\sum_{i=0}^N (n + i)^\alpha) = \Omega(N^{1+\alpha})$.
Since we can never be in credit debt,
we have that $n^{2-\gamma +(2-2\gamma)\alpha} \geq \Omega(N^{1+\alpha})$ and so $N\leq O(n^{\frac{2-\gamma +(2-2\gamma)\alpha}{1+\alpha}})$.

The number of cheaper insertions that we rolled back is $O(n^{1+\gamma})$. Each one of these costs at most $N^\alpha$. So the total time of the entire sequence of operations which solves \ThreeSUM{} is
\[
O(n^{1+\gamma}N^{\alpha} + N^{1+\alpha}) \leq O(n^{1+\gamma + \frac{(2-\gamma +(2-2\gamma)\alpha)\alpha}{1+\alpha}}+n^{2-\gamma +(2-2\gamma)\alpha}).
\]
The \ThreeSUM{} conjecture implies this bound must be $\Omega(n^{2-o(1)})$, so up to the $o(1)$ term we must have
\[
2 \leq  \max\left\{1+\gamma + \frac{(2-\gamma +(2-2\gamma)\alpha)\alpha}{1+\alpha}, \; 2-\gamma +(2-2\gamma)\alpha\right\}.
\]
The two terms are equal when $\gamma = \frac{1+\alpha}{2+3\alpha}$
and then $2\leq 2-\gamma + (2-2\gamma)\alpha$, implying that
\[
\alpha \geq \frac{\gamma}{2(1-\gamma)} = \frac{(1+\alpha)/(2+3\alpha)}{2(1+2\alpha)/(2+3\alpha)} = \frac{1+\alpha}{2+4\alpha}.
\]
Rearranging terms, we have $4\alpha^2 + \alpha -1 \ge 0$,
and we can set $\alpha$ to be $\frac{\sqrt {17} -1} 8 > 0.3903$.

\subsection{Dynamic MCM with Preprocessing --- Proof of Theorem~\ref{thm:dynamic-MCM-preprocessing}}\label{sec:MCM_prep}

The proof follows the same lines as the proof of Theorem~\ref{thm:MCM}. Let $n$ denote the size of the \ThreeSUM{} instance and
let $m$ denote the number of edges in the graph
on which we compute maximum cardinality matchings.

The initial graph created from an instance of the
offline \SetDisjointness{} problem is exactly the
same as the initial graph from Section~\ref{section:MCM}.
Recall that this graph has $m=\Theta(n^{2-\gamma})$ edges and
a \SetDisjointness{} query is implemented by adding two edges to the graph.
Since in the proof here we support fully dynamic graphs, this time in order to facilitate additional \SetDisjointness{} queries we undo the effect of adding the two edges by deleting both of these edges. Thus each \SetDisjointness{} query is processed by executing four edge updates, but not changing the number of vertices.

By Theorem~\ref{thm:improved_reduction}, assuming the \ThreeSUM{} conjecture, if the number of \SetDisjointness{} queries is $q=\Theta(n^{1+\gamma}\log n) = \Theta(m^{\frac{1+\gamma}{2-\gamma}}\log n)$ then either $t_p = \Omega(n^{2-o(1)})$ or $ q\cdot t_u = \Omega(n^{2-o(1)})$. Together this implies that
\begin{align*}
  t_p + m^{\frac{1+\gamma}{2-\gamma}}\cdot t_u & = t_p + q\cdot t_u \\
   & = \Omega(n^{2-o(1)}) \\
   & = \Omega(m^{\frac{2}{2-\gamma}-o(1)}).
\end{align*}

Finally, notice that providing a specific initial MCM does not assist in answering the \SetDisjointness{} queries, since the answer to the \SetDisjointness{} queries depends only on whether the MCM increases in size. Moreover, by the construction of the graph given in Section~\ref{section:MCM}, the longest length of an augmenting path encountered as edges are added and deleted is 7. Thus, the same lower bound holds for any algorithm for approximating MCM in a fully dynamic graph that reports the size of some matching without length-7 augmenting paths.

\section{Applications}\label{sec:applications}
\subsection{Online \SetIntersection{}}\label{app:details_set_intersection}

We consider \emph{online} \SetIntersection{} data structures that have a specified preprocessing time,
query time, and reporting time (per element).  By the trivial reduction from offline \SetIntersection{}
to online \SetIntersection, we discover tradeoffs between these three time bounds.

\begin{theorem}\label{thm:set_intersection_reporting_improved}
Assume the \ThreeSUM{} conjecture.  Fix any $\gamma \in [0,1)$, $\delta \in (0,1]$, and any online \SetIntersection{} data structure.
Let $t_p$ be its expected preprocessing time,
$t_q$ be its amortized expected query time, and
$t_r$ be its amortized expected reporting time per element.
Then
\[
t_p + N^{\frac{2(1+\gamma)}{3+\delta-\gamma}}\cdot t_q + N^{\frac{2(2-\delta)}{3+\delta-\gamma}}\cdot t_r = \Omega\left(N^{\frac{4}{3+\delta-\gamma}-o(1)}\right).
\]
\end{theorem}

\begin{proof}
We use the same reduction as the one in the proof of Theorem~\ref{thm:set_disjoint_improved}.
Using Theorem~\ref{thm:improved_reduction_reporting}, we have $N=\Theta(n^{1-\gamma} \sqrt{n^{1+\delta-\gamma}}) = \Theta(n^{\frac{3+\delta-\gamma}{2}})$, the number of queries is $\Theta(n^{1+\gamma}) = \Theta(N^{\frac{2(1+\gamma)}{3+\delta-\gamma}})$, and the total size of the output is $\Theta(n^{2-\delta}) = \Theta(N^{\frac{2(2-\delta)}{3+\delta-\gamma}})$. Thus, we obtain the following lower bound tradeoff:
\begin{align*}
t_p + N^{\frac{2(1+\gamma)}{3+\delta-\gamma}}\cdot t_q + N^{\frac{2(2-\delta)}{3+\delta-\gamma}} \cdot t_r & = \Omega(n^{2-o(1)}) =  \Omega\left(N^{\frac{4}{3+\delta-\gamma}-o(1)}\right).
\end{align*}
\end{proof}

Let us illustrate some implications of Theorem~\ref{thm:set_intersection_reporting_improved}.

\begin{corollary}
\label{cor:set_intersect_improved}
Assume the \ThreeSUM{} conjecture.
Fix constants $p\in [4/3,2)$ and $q\in [0,2/3]$
and suppose there is a data structure for \SetIntersection{} where
$t_p = O(N^{p}), t_q = O(N^{q}),$ and $t_r = N^{o(1)}$.  Then
\[
p +q \ge  2.
\]
\end{corollary}

\begin{proof}
Assume, for the purpose of obtaining a contradiction, that $p+q$ is
\emph{strictly} smaller than $2$, say $q = 2-p-\epsilon$ for some constant $\epsilon>0$.
By Theorem~\ref{thm:set_disjoint_improved}, for any constants $\gamma \in (0,1)$ and $\delta\in (0,1]$,

\begin{equation}
\frac{4}{3+\delta-\gamma}
    \le \max \left\{p,\; \frac{2(1+\gamma)}{3+\delta-\gamma}+q,\; \frac{2(2-\delta)}{3+\delta-\gamma}+o(1)\right\}
    = \max \left\{p,\; \frac{2(1+\gamma)}{3+\delta-\gamma}+2-p-\epsilon\right\}.\label{eqn:cor-2}
 \end{equation}

Observe that the maximum of the three expressions in (\ref{eqn:cor-2}) can never be equal to the third,
since $\frac{2(2-\delta)}{3+\delta-\gamma}+o(1) < \frac{4}{3+\delta-\gamma}$, which justifies
the last equality of (\ref{eqn:cor-2}).  We now want to choose $\gamma,\delta$ such that
the maximum of (\ref{eqn:cor-2}) is achieved in the second expression.

Since $p\in [4/3,2)$, we can always set
$\gamma\in [0,1)$ and $\delta\in (0,1]$ such that
\begin{equation}\label{eqn:cor-2b}
\gamma -\delta = 3 - \frac{4}{p} + \epsilon',
\end{equation}
for some constant $\epsilon'$ such that $0 < \delta \ll \epsilon' \ll \epsilon$.
E.g., when $p$ is close to 2 we set $\gamma$ to be slightly less than $1$, $\epsilon'$ to be slightly less than $4/p-2$,
and $\delta\ll \epsilon'$ appropriately.
Rewriting (\ref{eqn:cor-2b}), we have
\[
p = \frac{4-p\epsilon'}{3+\delta-\gamma} < \frac{4}{3+\delta-\gamma}.
\]
This implies that the maximum of (\ref{eqn:cor-2}) is attained in the
second expression, so
\[
\frac{4}{3+\delta-\gamma} \le \frac{2(1+\gamma)}{3+\delta-\gamma}+2-p-\epsilon.
\]

Rearranging terms, we have
\begin{align*}
\epsilon & \le \frac{2(1+\gamma)}{3+\delta-\gamma} +2-p - \frac{4}{3+\delta-\gamma}\\
      & = \zeromath{\displaystyle \frac{4+2\delta}{3+\delta-\gamma} -p}\hcm[2.6]=\;
      \frac{4+2\delta}{3-(3- \frac 4 p +\epsilon')} -p & \mbox{from (\ref{eqn:cor-2b})}\\
      & = \zeromath{\displaystyle \frac{4+2\delta}{\frac 4 p -\epsilon'} -p}\hcm[2.6]=\; 		
	p\cdot \left(\frac{4+2\delta}{4-p\epsilon'} - 1\right)\\
      & = \Theta(\epsilon') < \epsilon,			& \mbox{since $\delta < \epsilon'\ll\epsilon$ and $p\in[4/3,2)$}
\end{align*}
which is a contradiction, hence $p+q$ cannot be strictly smaller than 2.
\end{proof}

\subsection{Online \SetDisjointness{} --- Proof of Corollary~\ref{cor:set_disjoint_improved}}\label{sec:set_disjoint_improved_proof}

Recall that Corollary~\ref{cor:set_disjoint_improved} stated that if the preprocessing and
query times of a \SetDisjointness{} data structure were $O(N^p)$ and $O(N^q)$,
then the $\ThreeSUM$ conjecture implies $p+2q\ge 2$.  The proof follows the
same lines as that of Corollary~\ref{cor:set_intersect_improved}, but is somewhat simpler.

\begin{proof}
Assume, for the purpose of obtaining a contradiction, that $p +2q$ is
\emph{strictly} smaller than $2$, say $q = \frac{2-p}{2}-\epsilon$.
By Theorem~\ref{thm:set_disjoint_improved}, for any $\gamma \in (0,1)$,
\begin{equation}
\frac{2}{2-\gamma}
	\le  \max\left\{p, \; \frac{1+\gamma}{2-\gamma} + q\right\}
	= \max\left\{p, \; \frac{1+\gamma}{2-\gamma} + \frac{2-p}2-\epsilon\right\}.\label{eqn:cor-1}
 \end{equation}
Since $p\in [1,2)$, we can always set
$\gamma\in (0,1)$ such that
$\gamma = 2 - \frac{2}{p} + \epsilon'$, for some $\epsilon'$ such that $0 < \epsilon' \ll \epsilon$.
This implies $p=\frac{2-p\epsilon'}{2-\gamma} < \frac{2}{2-\gamma}$, which means
we can conclude the maximum of (\ref{eqn:cor-1}) must be attained in the
second expression, so
\[
\frac{2}{2-\gamma} \le \frac{1+\gamma}{2-\gamma}+\frac{2-p}{2}-\epsilon.
\]
Rearranging terms, we have
\begin{align*}
\epsilon & \le \frac{1+\gamma}{2-\gamma}+\frac{2-p} 2 - \frac{2}{2-\gamma}
	  \;=\; \frac{1}{2-\gamma}-\frac{p} 2
	  \;=\; \frac{1}{2-(2-\frac 2 p + \epsilon')}-\frac{p} 2
	  \;=\; p\left(\frac{1}{2-p\epsilon'} - \frac{1}{2}\right)\\
	 & = \Theta(\epsilon') < \epsilon,
\end{align*}
which is a contradiction, hence $p+2q$ cannot be strictly smaller than $2$.
\end{proof}

\subsection{$d$-Failure Connectivity}\label{section:d_failure_conn}

In the \textit{$d$-Failure Connectivity Oracle} problem we wish to preprocess an undirected graph
$G=(V,E)$ in order to support a \emph{single batch} of a set $F\subset V$ of $d$ vertex failures
and subsequent \emph{connectivity queries} in the subgraph induced by $V\backslash F$.

Duan and Pettie~\cite{DuanP10} introduced a $d$-failure connectivity structure whose
preprocessing and batch deletion times are $O(d^{1-2/c}mn^{1/c}\poly(\log n))$
and $O(d^{2c+4}\poly(\log n))$,
where $c\ge 1$ is an integer parameter.  The connectivity query time is $O(d)$, independent of $c,m,$ and $n$.
The same authors presented a different $d$-failure connectivity structure~\cite{DuanP17}
with $O(mn\log n)$ preprocessing time, $O(d^2\poly(\log n))$ batch deletion time, and the same $O(d)$ query time.
The main open question in this line of work is whether the $O(d)$ query time of~\cite{DuanP10,DuanP17}
could be improved to match the
$d$-\emph{edge} failure connectivity oracles~\cite{PatrascuT07,KapronKM13,DuanP10,DuanP17},
whose query time is $\tilde{O}(1)$, independent of $d$.
Here we prove that with preprocessing and batch deletion times similar to~\cite{DuanP10},
the query time must depend on $d$, and be $\Omega(d^{1/2-o(1)})$.\footnote{Our lower bound does not apply
to~\cite{DuanP17}, whose preprocessing time is quadratic, and therefore already large enough to solve the underlying
\ThreeSUM{} instance.}  Subsequent to the initial publication of this work~\cite{KopelowitzPP16},
Henzinger et al.~\cite{HenzingerKNS15} proved that, conditioned on the \OMv{} conjecture,
$d$-failure connectivity oracles with $\poly(n)$ preprocessing time and reasonable batch deletion times
require $\Omega(d^{1-o(1)})$ query times.  Some recent upper
bounds~\cite{LarsenW17,ChakrabortyKL18} have cast some doubt on the validity of the \OMv{} conjecture.

\begin{theorem}\label{thm:d_fail}
Assume the \ThreeSUM{} conjecture.
For any $1/2\leq \gamma <1$ suppose there is a $d$-failure connectivity structure for
$d^{\frac{2-\gamma}{2-2\gamma}}$-edge, $d^{\frac{1}{2-2\gamma}}$-vertex graphs
with expected preprocessing time $t_{p}$, expected deletion time $t_{d}$, and expected query time $t_{q}$. Then,
\[
t_p + d^{\frac{1}{2-2\gamma}}\cdot t_d + d^{\frac{1+\gamma}{2-2\gamma}}\cdot t_q = \Omega(d^{\frac 1 {1-\gamma}-o(1)}).
\]
\end{theorem}

\begin{proof}
We reduce the \SetDisjointness{} problem to the $d$-failure connectivity as follows.
We make use of Theorem~\ref{thm:improved_reduction} and set $d = |U| = O(n^{2-2\gamma})$.
Construct a tripartite graph $G=(V,E)$ on vertices $V = A\cup B \cup U$, where $A$ and $B$ are copies of $\mathcal F$, and edges
\[
E = \{(a,c) \;|\; c\in a\} \;\cup\; \{(b,c) \;|\; c\in b\}
\]
We now need to answer $n^{1+\gamma} = O(d^{\frac{1+\gamma}{2-2\gamma}})$ \SetDisjointness{} queries using a black-box data structure
for $d$-failure connectivity on $G$.  For each $a\in A$ separately we perform up to $d$ deletions and then answer {\em all}
\SetDisjointness{} queries involving $a$ using connectivity queries.
To do this we delete all vertices in $U$ that correspond to elements not in $a$ and let $G[a]$ be the resulting graph.
Notice that in $G[a]$, $a$ is only connected to sets in $A\cup B$ that intersect $a$.  We can therefore answer
any \SetDisjointness{} query ``$a\cap b=\emptyset$?'' by asking one connectivity query in $G[a]$.

Observe that $G$ is an $M$-edge, $N$-vertex graph where
$N = |\mathcal{F}| + |U|= O(n\log n) = \tilde{O}(d^{\frac{1}{2-2\gamma}})$ and
$M=O(n^{2-\gamma}) = O(d^{\frac{2-\gamma}{2-2\gamma}})$.
Thus, $t_p + (n\log n)\cdot t_d + n^{1+\gamma}\cdot t_q = \Omega(n^{2-o(1)})$.
Substituting $n = \Omega(d^{\frac{1}{2-2\gamma}})$ completes the proof.
\end{proof}

Note that $\gamma$ does not affect the final conclusion that connectivity queries require $\Omega(d^{1/2-o(1)})$ time.
The role of the $\gamma$ parameter is to make the total time for all batch deletions negligible.  For example,
if $t_d = d^{100}$, we would have to set $\gamma$ very close to 1 so that $d^{\frac{1}{2-2\gamma}}\cdot d^{100} \ll d^{\frac{1}{1-\gamma}}$.

\subsection{Document Retrieval Problems with Multiple Patterns}

One of the services offered by search engines is the retrieval of documents
whose text satisfies some predicate, typically the inclusion (or exclusion) of multiple keywords.
In this section we prove lower bounds on several problems of this type.

\subsubsection{Two Pattern Document Retrieval}

In the \textit{Document Retrieval} problem~\cite{Muthukrishnan02} we are interested in preprocessing a corpus of
documents $X=\{D_1,\cdots, D_k\}$ where $N=\sum_{D\in X} |D|$,
so that given a pattern $P$ we can quickly report all of the documents that
contain $P$. We are usually interested in run times that depend on the number of documents that contain
$P$, not on the total number of occurrences of $P$ in the entire corpus.
In the \textit{Two Pattern Document Retrieval} problem we are given two patterns
$P_1$ and $P_2$ at query time, and wish to report all of the documents that contain both $P_1$ and $P_2$.
We consider two versions of the Two Pattern Document Retrieval problem.
In the reporting version we are interested in enumerating all documents that contain both patterns.
In the decision version we only want to decide whether the output is non-empty or not.

All known solutions for the Two Pattern Document Retrieval problem with non
trivial preprocessing use at least $\tilde{\Omega}(\sqrt N)$ time per
query~\cite{Muthukrishnan02,CP10,HSTV10,HSTV12,KPP15}.
Larsen, Munro, Nielsen, and Thankachan~\cite{LMNT14} prove lower bounds
on Two Pattern Document Retrieval, conditioned on the hardness of combinatorial boolean matrix multiplication.
(A data structure with $N^{3/2-\epsilon}$ preprocessing and $(\sqrt{N})^{1-\epsilon}$-time queries implies
a subcubic combinatorial BMM algorithm.)
We provide some additional evidence of hardness conditioned on the \ThreeSUM{} conjecture.

It is straightforward to see that the two versions of Two Pattern Document Retrieval
solve \SetIntersection{} and \SetDisjointness, respectively.
In particular, the reduction creates an alphabet $\Sigma$ which corresponds to all of the
sets in $\mathcal{F}$. For each $e\in U$ we create a document that contains the characters corresponding
to the sets that contain $e$. The intersection between $S,S'\in \mathcal F$ directly corresponds to
all the documents that contain both symbols $S$ and $S'$. Thus, all of the lower bound tradeoffs for intersection
problems are lower bound tradeoffs for the Two Pattern Document Retrieval problem.

\begin{theorem}\label{thm:two_pat_doc_retrieval}
Assume the \ThreeSUM{} conjecture.
Fix any $\gamma \in [0,1)$, and consider any data structure for Two Pattern Document Retrieval for a corpus $X$
with expected preprocessing time $t_p$ and query time $t_q$.  These time bounds are a function of
$N=\sum_{D\in X}|D|$, the size of the corpus.  Then
\[
t_p + N^{\frac{1+\gamma}{2-\gamma}} \cdot t_q = \Omega\left(N^{\frac{2}{2-\gamma}-o(1)}\right).
\]
\end{theorem}

\begin{theorem}\label{thm:two_pat_doc_retrieval_reporting}
Assume the \ThreeSUM{} conjecture.
Fix any $\gamma \in [0,1)$ and $\delta \in (0,1)$, and consider any data structure for Two Pattern Document Retrieval for a corpus $X$
with expected preprocessing time $t_p$, query time $t_q$, and reporting time (per document) $t_r$.
These time bounds are a function of $N=\sum_{D\in X}|D|$, the size of the corpus.  Then
\[
t_p + N^{\frac{2(1+\gamma)}{3+\delta-\gamma}}\cdot t_q + N^{\frac{2(2-\delta)}{3+\delta-\gamma}}\cdot t_r = \Omega\left(N^{\frac{4}{3+\delta-\gamma}-o(1)}\right).
\]

\end{theorem}

\subsubsection{Forbidden Pattern Document Retrieval}

In the \textit{Forbidden Pattern Document Retrieval} problem~\cite{FGKLMSV12} we are still interested in preprocessing
a fixed document corpus.  A query now consists of two patterns $P^+,P^-$ and must report all of the documents that
contain $P^+$ and do not contain $P^-$ (reporting version), or decide whether there exists no such document (decision version).

All known solutions for the Forbidden Pattern Document Retrieval problem with non trivial preprocessing use at least $\Omega(\sqrt N)$ time per query~\cite{FGKLMSV12,HSTV12}.
Larsen, Munro, Nielsen, and Thankachan~\cite{LMNT14}
also give lower bounds on this problem, conditioned on the combinatorial BMM conjecture.
Here we provide some additional evidence of hardness conditioned on the \ThreeSUM{} conjecture.

\begin{theorem}\label{thm:forbid_pat_doc_retrieval}
Assume the \ThreeSUM{} conjecture.
For any $\gamma \in [0,1)$ and any data structure for Forbidden Pattern Document Retrieval for a corpus $X$
with expected preprocessing time $t_p$ and expected query time $t_q$, which depend on $N=\sum_{D\in X}|D|$.
Then
\[
t_p + N^{\frac{1+\gamma}{3-2\gamma}} \cdot t_q  = \Omega(N^{\frac{2}{3-2\gamma}-o(1)}).
\]

\end{theorem}

\begin{proof}
We create two copies of $mathcal F$, denoted by $A$ and $B$, and similar to the proof of Theorem~\ref{thm:two_pat_doc_retrieval} we set $\Sigma = A\cup B$. For each $e\in U$ we create a document that contains all of the characters corresponding to sets from $A$ that contain $c$ and sets from $B$ that do not contain $c$.

Using Theorem~\ref{thm:improved_reduction}, we have $N = \Theta(n^{3-2\gamma})$,
 and the number of queries to answer is $\Theta(n^{1+\gamma} ) = \Theta(N^{\frac{1+\gamma}{3-2\gamma}})$.
 Thus we obtain the following lower bound tradeoff:
\[
 t_p + N^{\frac{1+\gamma}{3-2\gamma}} \cdot t_q  = \Omega(n^{2-o(1)}) = \Omega(N^{\frac{2}{3-2\gamma}-o(1)}).
\]
\end{proof}

Notice that
if we only allow linear preprocessing time then by making $\gamma$ arbitrarily small we obtain a
query time lower bound of $\Omega(N^{\frac{1}{3}-o(1)})$.

\begin{theorem}\label{thm:forbid_pat_doc_retrieval_reporting}
Assume the \ThreeSUM{} conjecture.
Fix any $\gamma \in [0,1), \delta\in(0,1)$, and any data structure for the reporting version of
Forbidden Pattern Document Retrieval for a corpus $X$,
with expected preprocessing time $t_p$, expected query time $t_q$,
and amortized expected reporting time $t_r$ (per document), which depend on $N=\sum_{D\in X}|D|$.
Then
\[
t_p +
N^{\frac{1+\gamma}{\frac{3}{2}(1+\delta-\frac{\gamma}{3})}}\cdot t_q +
N^{\frac{2-\delta}{\frac{3}{2}(1+\delta-\frac{\gamma}{3})}}\cdot t_r = \Omega(N^{\frac{2}{\frac{3}{2}(1+\delta-\frac{\gamma}{3})}-o(1)}).
\]
\end{theorem}

\begin{proof}
Our proof is similar to the proof of Theorem~\ref{thm:forbid_pat_doc_retrieval}, only this time we use
Theorem~\ref{thm:improved_reduction_reporting}.
So we have
$N=\Theta(n^{1+\delta-\gamma}\sqrt{n^{1+\delta+\gamma}}) = \Theta(n^{\frac{3}{2}(1+\delta-\frac{\gamma}{3})})$,
the number of queries is $\Theta(n^{1+\gamma}) = \Theta(N^{\frac{1+\gamma}{\frac{3}{2}(1+\delta-\frac{\gamma}{3})}})$,
and the total size of the output is $\Theta(n^{2-\delta}) = \Theta(N^{\frac{2-\delta}{\frac{3}{2}(1+\delta-\frac{\gamma}{3})}})$.
Thus, we obtain the following lower bound tradeoff:
\[
t_p
+ N^{\frac{1+\gamma}{\frac{3}{2}(1+\delta-\frac{\gamma}{3})}}\cdot t_q
+ N^{\frac{2-\delta}{\frac{3}{2}(1+\delta-\frac{\gamma}{3})}}\cdot t_r
= \Omega(n^{2-o(1)}) = \Omega(N^{\frac{2}{\frac{3}{2}(1+\delta-\frac{\gamma}{3})}-o(1)}).
\]
\end{proof}

Notice that allowing only linear preprocessing time and a constant reporting time, then by making
$\gamma$ and $\delta$ arbitrarily small we obtain a query time lower bound of $\Omega(N^{\frac{2}{3}-o(1)})$.

\section{\ConvThreeSUM{} vs. \ThreeSUM{} --- Proof of Theorem~\ref{thm:conv3sum}}\label{app:proof_conv3sum}

In this section we can tolerate approximately universal hash functions, and a
more natural definition of almost linearity.

\begin{definition} {\bf (Universality and Linearity)}
Let $\mathcal{H}$ be a family of hash functions from $[u] \rightarrow [m]$.
\begin{enumerate}
\item $\mathcal{H}$ is called {\em $c$-universal} if for any distinct $x,x' \in [u]$,
\[
\Pr_{h\in \mathcal{H}}(h(x) = h(x')) \le \frac{c}{m}.
\]

\item $\mathcal{H}$ is called {\em almost linear} if for any $h\in\mathcal{H}$ and any $x,x' \in [u]$,
\[
h(x) + h(x') \equiv h(x+x') + \{-1, 0\} \; (\modulo m).
\]
\end{enumerate}
\end{definition}

By tolerating approximate universality, we can use the simple hash functions
analyzed by Dietzfelbinger et al.~\cite{DietzfelbingerHKP97}.

\begin{theorem}\label{thm:Dietzfelbinger} (Dietzfelbinger, Hagerup, Katajainen, and Penttonen~\cite{DietzfelbingerHKP97})
Let $u$ and $m$ be powers of two, with $m<u$.
The family $\mathcal{H}_{u,m}$ is 2-universal and almost linear, where
\begin{align*}
\mathcal{H}_{u,m} &= \left\{h_a : [u]\rightarrow[m] \;\,|\,\;
	  \mbox{$a\in [u]$ is an odd integer}
  \right\}\\
\mbox{and } \; h_{a}(x) &= (ax \modulo u)\div (u/m).
\end{align*}
\end{theorem}

Because the modular arithmetic and division are by powers of two,
the hash functions of Theorem~\ref{thm:Dietzfelbinger}
are very easy to implement using standard multiplication and shifts.
If $u=2^w$, where $w$ is the number of bits per word, and $m = 2^s$,
the function is written in C as \texttt{(a*x) >> (w-s)}.
Dietzfelbinger et al.~\cite{DietzfelbingerHKP97} proved that it is 2-universal.
It is clearly almost linear.

\subsection{Hashing and Coding Preliminaries}\label{sect:hashing}

The reduction in the next section makes use of any constant rate, constant relative distance binary code.
The expander codes of Sipser and Spielman~\cite{SipserS96} are sufficient for our application.

\begin{theorem}\label{thm:expandercode} (See Sipser and Spielman~\cite{SipserS96})
There is a constant $\epsilon>0$ such that for any sufficiently large $\delta > \delta(\epsilon)$,
there is a binary code $C : \{0,1\}^{N} \rightarrow \{0,1\}^{\delta N}$ such that
for any $x,y\in\{0,1\}^N$, the Hamming distance between $C(x)$ and $C(y)$ is at least $\epsilon\cdot \delta N$.
Moreover, $C(x)$ can be computed in $O(\delta N)$ time.
\end{theorem}

\subsection{The Reduction}\label{sect:reduction}

Let $[u]\backslash\{0\} = [2^w]\backslash\{0\}$ be the universe.  It is convenient to assume that $0$ is excluded from $A$, but this
is without loss of generality since all witnesses involving $0$ can be enumerated in $O(n\log n)$ time by sorting $A$.
Choose $L$ hash functions $(h_i)_{i\in [L]}$ independently from $\mathcal{H}_{u,m}$, where
$m = 2^{\ceil{\log n}}$ is the least power of two larger than $n$.
Ideally a hash function
will map $A$ injectively into the buckets $[m]$, or at least put a constant load on each bucket, but this cannot be guaranteed.
Some buckets will be overloaded and the items in them discarded.

\begin{definition} {\bf (Overloaded Buckets, Discarded Elements)}
For each $i\in [L]$ and $j\in [m]$ define
\begin{align*}
\bucket_i(j) &= \{x \in A \:|\: h_i(x) = j\}
\intertext{to be the set of elements hashed by $h_i$ to the $j$th bucket.  The truncation of this bucket is defined as}
\bbucket_i(j) &= \left\{\begin{array}{ll}
\bucket_i(j)		& \mbox{if $|\bucket_i(j)| \le T$}\\
\emptyset			& \mbox{otherwise}
\end{array}\right.
\end{align*}
where $T=O(1)$ is a constant threshold to be determined.  If $\bbucket_i(j) = \emptyset$ we say
that the elements of $\bucket_i(j)$ were {\em discarded by $h_i$}.
An element is called {\em bad} if it is discarded by a $4/T$-fraction of the hash functions.
\end{definition}

\begin{lemma}
The probability that an element is bad is at most $\exp\left(-\frac{2L}{3T}\right)$.
\end{lemma}

\begin{proof}
Since each $h_i$ is $2$-universal,
the expected number of other elements in $x$'s bucket is, by linearity of expectation,
at most $2(n-1)/m < 2$.  By Markov's inequality the probability that $x$ is discarded by $h_i$ is less than $2/T$.
Let $X$ be the number of hash functions that discard $x$, so $\E(X) < 2L/T$. By definition $x$ is bad
if $X > 4L/T > 2\cdot \E(X)$.  Since the hash functions were chosen independently, by a Chernoff
bound, $\Pr(x \mbox{ is bad}) < \exp\left(-\frac{2L}{3T}\right)$.
\end{proof}

We will set $T=O(1)$ and $L=\Theta(\log n)$ to be sufficiently large so that the probability
that no elements are bad is $1-1/\poly(n)$.
We proceed under the assumption that there are no bad elements.

\begin{lemma}\label{lem:three-not-discarded}
Suppose there are no bad elements with respect to $(h_i)_{i\in [L]}$.
For any three $a,b,c \in A$,
there are more than $\left(1 - \frac{12}{T}\right)L$ indices $i\in[L]$ such that $h_i$ discards none of $\{a,b,c\}$.
\end{lemma}

\begin{proof}
Each of $a,b,c$ is discarded by less than $4L/T$ hash functions, so none are discarded by
at least $L - 12L/T$ hash functions.
\end{proof}

Let $\delta>1,\epsilon>0$ be the parameters of Theorem~\ref{thm:expandercode},
where $N = \ceil{\log n}$ and $L=\delta N$.
We assign each $x\in A$ an $L$-bit codeword $C_x$ such
that any two $C_x,C_y$ disagree in at least $\epsilon L$ positions.

We make
$8TL$ calls to a \ConvolutionThreeSUM{} algorithm on vectors $\{A_{\ell}\}_{\ell\in[L]\times \{-1,0\} \times\{0,1\}\times [2T]}$,
each of length $14m = O(n)$.
For reasons that will become clear we index the calls by tuples $\ell = (i,\alpha,\beta,\gamma) \in [L]\times \{-1,0\} \times\{0,1\}\times [2T]$.
The first coordinate $i$ of $\ell$ identifies the hash function.
The second coordinate $\alpha$
indicates that we are looking for witnesses $a,b,a+b\in A$ for which $h_i(a) + h_i(b) = h_i(a+b) + \alpha \, (\modulo m)$.
A natural way to define $A_{\ell}$ creates multiple copies of elements but can lead to a situation where there are false positives:
we may have $A_\ell(p) + A_\ell(q) = A_\ell(p+q)$ and yet this is not a witness for the original \ThreeSUM{} instance because
$A_\ell(p) = A_\ell(q)$.\footnote{This minor bug appears in \Patrascu{}'s reduction from \ThreeSUM{} to \ConvolutionThreeSUM.}
In each call to \ConvolutionThreeSUM{} we look for witnesses where each element can play the role of
either ``$p$'' or ``$q$'' in the example above,
but not both; all elements will be eligible to play the role of ``$p+q$.''
The parity of $C_x(i) \mbox{ \sc xor } \beta$ tells us which roles $x$ is allowed to play, where $\beta$ is the third coordinate of $\ell$.
The fourth coordinate $\gamma$ of $\ell$ effects a cyclic shift of the order of elements within a bucket.

Each vector $A_{\ell}$ is partitioned into $2m$ contiguous {\em blocks}, each of length $7T$.
Many of the locations of $A_{\ell}$ are filled with a dummy value $\infty$, which is some sufficiently large number
that cannot be part of any witness, say $2\max(A)+1$.
The elements of the $j$th bucket each appear three times in $A_\ell$, twice in the first half and once in the second.

Order the elements of $\bbucket_i(j)$ arbitrarily as $(x(i,j,k))_{k\in [T]}$, where $x(i,j,k)$ does not exist if $k\ge |\bbucket_i(j)|$.
Define the vector $A_{(i,\alpha,\beta,\gamma)}$ as follows.
\begin{align*}
A_{(i,\alpha,\beta,\gamma)}(j(7T) + t)
&= \left\{
\begin{array}{ll}
x(i,j, k)		&  \mbox{if $t = T + k$, $k\in [T]$, and $C_{x(i,j, k)}(i) \mbox{ \sc xor } \beta = 0$}\\
x(i,j, k)		& \mbox{if $t = 2T + k$, $k\in [T]$, and $C_{x(i,j, k)}(i) \mbox{ \sc xor } \beta = 1$}\\
x(i,(j-\alpha)\modulo m,k)   & \mbox{if $t = 3T + ((k+ \gamma) \modulo 2T)$ and $k\in [T]$}\\
\infty			& \mbox{otherwise.}
\end{array}
\right.
\end{align*}
The last case applies when $j, k,$ or $t$ is out of range
or if the given element, say $x(i,j,k)$, does not exist because $|\bbucket_i(j)| \le k$.
See Figure~\ref{fig:block}.

\begin{figure*}
\centering
\scalebox{.37}{\includegraphics{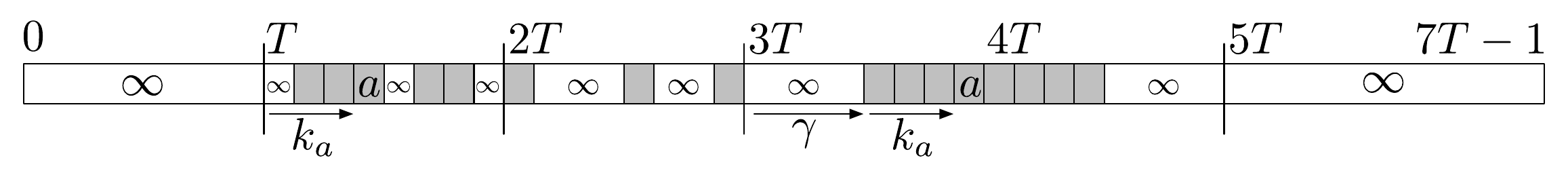}}
\caption{\label{fig:block}Block $j$ in $A_\ell$ occupies positions $j(7T)$ through $(j+1)(7T)-1$.
In the first half of $A_\ell$, a block is partitioned into five intervals.  The first interval covers positions
$0$ through $T-1$ and is always filled with a dummy value $\infty$.
The second and third intervals run, respectively, from positions $T$ through $2T-1$ and positions $2T$ through $3T-1$.
They contain those elements $x\in \bbucket_i(j-\alpha)$ for which $C_x(i) \mbox{ \sc xor } \beta$ is, respectively, $0$ and $1$.
The fourth interval runs from positions $3T$ through $5T-1$ and contains {\em all} members of $\bbucket_i(j-\alpha)$, cyclically
shifted by $\gamma$.  The last interval, from positions $5T$ through $7T-1$, is always filled with dummies.
The composition of a block $j$ in the {\em second} half of $A_\ell$ is similar,
except that the second and third intervals (positions $T$ through $3T-1$) contain only dummies,
and the fourth interval contains all members of $\bbucket_i((j-\alpha)\, \modulo m)$.
}
\end{figure*}

\begin{lemma} {\bf (No False Negatives)}
Suppose $a,b,a+b\in A$ is a witness to the \ThreeSUM{} instance $A$.
For some $\ell = (i,\alpha,\beta,\gamma)$,
this is also a witness in the \ConvolutionThreeSUM{} instance $A_{\ell}$.
\end{lemma}

\begin{proof}
Set the threshold $T = 12/\epsilon = O(1)$.  By Lemma~\ref{lem:three-not-discarded}
there are more than $L(1-12/T) = L(1-\epsilon)$ indices $i\in [L]$ such that none of $\{a,b,a+b\}$
are discarded by $h_i$.  Moreover, by the properties of the error correcting code (Theorem~\ref{thm:expandercode})
there are at least $\epsilon L$ indices $i$ for which
$C_a(i) \neq C_b(i)$, which implies that both criteria are satisfied for at least one $i$.
Fix any such $i$.

Let $j_a=h_i(a), j_b = h_i(b)$, and $j_{a+b} = h_i(a+b)$ be the bucket indices of $a,b,$ and $a+b$.
Let $k_a,k_b,k_{a+b}$ be their positions in those buckets, that is, $a = x(i,j_a,k_a)$ and $b = x(i,j_b,k_b)$,
and  $a+b = x(i,j_{a+b},k_{a+b})$.
Without loss of generality $j_a \le j_b$.
Let $\beta = C_a(i)$, so $C_a(i) \mbox{ \sc xor } \beta = 0$ and
$C_b(i) \mbox{ \sc xor } \beta = 1$.
Let $\alpha \in \{-1,0\}$ be such that
$h_i(a) + h_i(b) \equiv h_i(a+b) + \alpha \, (\modulo m)$.

In the vector $A_{(i,\alpha,\beta,\gamma)}$,
\begin{itemize}
\item $a$ is at position $j_a(7T) + T + k_a$, because $C_a(i) \mbox{ \sc xor } \beta = 0$,
\item $b$ is at position $j_b(7T) + 2T + k_b$, because $C_b(i) \mbox{ \sc xor } \beta = 1$,
\item and since $j_{a+b} \equiv j_a + j_b  - \alpha \, (\modulo m)$, $a+b$ is at position $(j_a + j_b)(7T) + 3T + ((k_{a+b} + \gamma) \modulo 2T)$.
\end{itemize}
Thus, for $\gamma = (k_a + k_b - k_{a+b}) \modulo 2T$,
the triple $(a,b,a+b)$ forms a witness for the \ConvolutionThreeSUM{} vector $A_\ell$.
\end{proof}

\begin{lemma} {\bf (No False Positives)}
If $(a,b,a+b)$ is a witness in some \ConvolutionThreeSUM{} instance $A_\ell$,
it is also a witness in the original \ThreeSUM{} instance $A$.
\end{lemma}

\begin{proof}
None of $\{a,b,a+b\}$ can be the dummy $\infty$ in $A_\ell$, so they must all be members of $A$.
The only way it cannot be an witness for \ThreeSUM{} is if $b=a$, that is, $(a,a,2a)$ is not a triple
of distinct numbers.  If $a$ is not discarded, it appears at exactly three positions in $A_\ell$.  Regardless of
the bit $C_a(i)$, $a$ appears at both
$A_{\ell}((j_a + \alpha)(7T) + 3T + ((k_a+\gamma) \modulo 2T))$
and $ A_\ell((m + j_a + \alpha)(7T) + 3T + ((k_a+\gamma) \modulo 2T))$
 for some $k_a \in [T] \mbox{ and } \gamma\in [2T]$.
Depending on the parity of $C_x(i) \mbox{ \sc xor } \beta$, $a$ also appears at either
$A_{\ell}(j_a(7T) + T + k_a)$
or $A_{\ell}(j_a(7T) + 2T + k_a)$.
For $(a,a,2a)$ to be a \ConvolutionThreeSUM{} witness we would need $2a$ to appear either at
$A_\ell\big((2j_a+\alpha)(7T) + 4T + k_a + ((k_a+\gamma) \modulo 2T)\big)$
or $A_\ell\big((2j_a+\alpha)(7T) + 5T + k_a + ((k_a+\gamma) \modulo 2T)\big)$.
However, in both of those positions $A_\ell$ is $\infty$ by definition.
See Figure~\ref{fig:block}.
\end{proof}


We have shown that the randomized (Las Vegas) complexities of \ThreeSUM{} and \ConvolutionThreeSUM{}
are equivalent up to a logarithmic factor.  Since hashing plays such an essential role in the reduction,
it would be surprising if our construction could be efficiently derandomized, or if it could be generalized to
show that \ThreeSUM{} and \ConvolutionThreeSUM{} over the reals are essentially equivalent.

The $O(\log n)$-factor gap in Theorem~\ref{thm:conv3sum} stems from our solution to two technical difficulties,
(i) ensuring that all triples appear in lightly loaded buckets with respect to a large fraction of the hash functions,
and (ii) ensuring that no non-\ThreeSUM{} witnesses $(a,a,2a)$ occur as witnesses in any \ConvolutionThreeSUM{} instance.
We leave it as an open problem to show that \ThreeSUM{} and \ConvolutionThreeSUM{} are asymptotically
equivalent, without the $O(\log n)$-factor gap.

\bibliographystyle{plain}
\bibliography{tsvi}

\end{document}